\documentclass[10pt]{article}

\usepackage[paperwidth=8.5in,paperheight=11.00in,top=1.00in, bottom=1.00in, left=0.75in, right=0.75in]{geometry}
\usepackage{amsmath, amsthm, amssymb}
\usepackage{geometry,fancyhdr}
\usepackage{graphicx}
\usepackage{dsfont}
\usepackage{mathtools}
\mathtoolsset{showonlyrefs=true}
\usepackage{thmtools,thm-restate}
\usepackage{bbm}
\usepackage{array}
\usepackage[dvipsnames]{xcolor}							% gives you color definitions
\usepackage{hyperref}
\hypersetup{
		colorlinks   = true,
		citecolor    = RoyalBlue, %NavyBlue,
		linkcolor    = RubineRed, %Maroon,
		urlcolor     = Turquoise
}
\usepackage{caption}
\usepackage{subcaption}

\newtheorem{lemma}{Lemma}
\newtheorem{prop}{Proposition}

\newcommand{\dd}{\mathrm{d}}
\newcommand{\bbE}{\mathds{E}}
\newcommand{\bbP}{\mathds{P}}
\newcommand{\bbR}{\mathds{R}}
\newcommand{\bbF}{\mathds{F}}

\newcommand{\bbL}{\mathcal{L}}
\newcommand{\bbT}{\mathcal{T}}

\newcommand{\bbC}{\mathcal{C}}

\usepackage{natbib}
 \bibpunct[, ]{(}{)}{,}{a}{}{,}%

\begin{document}

\title{Rooftop and Community Solar Adoption with Income Heterogeneity}

\author{
Swapnil Rayal
\thanks{Foster School of Business, University of Washington. \textbf{e-mail}: \url{srayal@uw.edu}}
\and 
Apurva Jain
\thanks{Foster School of Business, University of Washington.  \textbf{e-mail}: \url{apurva@uw.edu}}
\and
Matthew Lorig
\thanks{Department of Applied Mathematics, University of Washington.  \textbf{e-mail}: \url{mlorig@uw.edu}}
}

\maketitle

\begin{abstract}
Each household in a population characterized by income heterogeneity faces random demand for electricity and decides if and when it should adopt a solar product, rooftop solar or community solar. A central planner, aiming to meet an adoption level target within a set time, offers net metering and subsidy on solar products and minimizes its total cost. Our focus is on analyzing the interactions of three new features we add to the literature: income diversity, availability of community solar, and consideration of adoption timing. \textit{Methodology and results:} We develop a bilevel optimization formulation to derive the optimal subsidy policy. The upper level (planner's) problem is a constrained non-linear optimization model in which the planner aims to minimize the average subsidy cost. The lower level (household's) problem is an optimal stopping formulation, which captures the adoption decisions of the households. We derive a closed-form expression for the distribution of optimal adoption time of households for a given subsidy policy.  We show that the planner's problem is convex in the case of homogeneous subsidy for the two products. \textit{Managerial implications:} Our results underscore the importance for planners to consider three factors - adoption level target, time target, and subsidy budget - simultaneously as they work in tandem to influence the adoption outcome. The planners must also consider the inclusion of community solar in their plans because, as we show, community and rooftop solar attract households from different sides of the income spectrum. In the presence of income inequality, the availability of community makes it easier to meet solar adoption targets. 
\end{abstract}

\section{Introduction.}\label{intro}
Rooftop solar panels are a major component of governments' renewable energy strategies. State and local governments worldwide offer subsidy programs to promote the adoption of solar panels among households. The earliest such laws in the U.S. date to the 1980s (\cite{verzola2016crossing}). The popularity of tools like \href{https://sunroof.withgoogle.com/}{Google Project Sunroof} shows widespread awareness and interest in household rooftop solar adoptions. As such programs proliferate and mature, new concerns about their impact and effectiveness have emerged. In addition to tracking how many households adopt rooftop solar, the metrics of interest have broadened to include how fast such adoptions occur and how well-off these adopters are. Such concerns about timing and equity raise interesting questions regarding the design of subsidy programs and how households of different income profiles respond to them (\cite{WSJ2022}).\par
With global warming's effects ever more visible, there is an increasing pressure on government planners to specify time targets for the adoption-quantity goals they set. According to the National Conference of State Legislatures (\cite{NCSL2021}),  more than half of American states have set renewable energy targets, with ten aiming to achieve full conversion with deadlines ranging between 2030 and 2050. Some regions have set specific targets for solar energy;  for instance, Washington D.C. aims at a 5\% solar contribution in its complete transition to cleaner resources by 2032. The industry has the same focus; the Solar Energy Industries Association aims to reach 30\% of U.S. electricity generation by 2030 (\cite{SEIA2021}).\par
The concerns about equity are also gaining prominence. As activists point out, one problem with the currently implemented solar subsidy policies is that they benefit the higher-income households more than the low-income households. For example, the federal tax credit for solar panel installation cannot be claimed by a household that rents an apartment and cannot put a solar array on the roof. At the same time, it is easier for a homeowner to install a solar array and qualify for the tax credit (\cite{DOE2021}). There is evidence of disparity in solar adoption based on income level, race, and ethnicity (\cite{forrester2022residential}, \cite{sunter2019disparities}). Generally, solar adopters tend to be from higher-income households living in neighborhoods with a relatively higher proportion of non-Hispanic white and Asian populations (\cite{barbose2021residential}). \par
The recent emergence of community or subscription solar products may be able to address the concerns about speed and equity raised above.  These products allow households to lease or subscribe to a portion of shared solar systems located away from homes, thus removing the disadvantages of high upfront cost and space constraints in rooftop solar installation.  As of 2021, there are about 1600 community solar projects in 39 states, and 22 states have enacted community solar legislation (\cite{heeter2021sharing}).  Many states, including Delaware, New Hampshire, Oregon, and Minnesota, are crediting surplus community solar generation at retail price (\cite{DSIRE2021}), the same way they support rooftop solar products. The availability of rooftop solar and community solar products in a region gives the households a flexibility in payment schemes which can hasten their adoption decision.  Thus, community and subscription solar products may provide planners with an additional resource to promote the adoption of solar technology.\par
This paper aims to develop a model that can incorporate important but previously unaddressed features discussed above - the timing of adoption, the impact of income inequality, and the availability of community solar - while maintaining tractability. Models that analyze the adoption timing of a new technology are an integral part of Operations Management(OM) research, but there has been limited work on their application to solar products. In addition, the interplay between household adoption timing decisions and the planners' time targets has not yet been explored in the literature. The modeling of income inequality also remains an open question in the OM literature; we take an early step in that direction. Finally, we believe that this is the first effort to model a household's choice between two different solar products, rooftop and community solar. In the following, we describe the model, the types of questions we ask, and the insights we draw from it.\par
We model a regional population where the households differ in income and electricity consumption. Drawing on energy economics literature, we argue that income diversity correlates with differences in electricity demand and determines a household's propensity to invest in renewable energy. More specifically, a higher-income household consumes more electricity at a faster growth rate and, financially, finds it easier to make a fixed investment; we capture this by modeling a lower discount rate for a higher-income household. In other words, a lower-income household uses a higher discount rate to value the future benefits of renewable energy. See Section \ref{model} for detailed support of these modeling choices.\par
A household's electricity demand is modeled as a continuous time stochastic process. Each household solves an optimal stopping problem to decide when to adopt a solar product, either rooftop or community solar. The household incurs an electricity consumption cost before adoption; must spend a fixed amount for rooftop solar, or pay a subscription fee for community solar adoption; and enjoys a net-metering billing system after adoption where it can sell any excess capacity back to the grid. Our modeling choices are informed by practice and earlier literature and are further explained in Section \ref{model}.\par
The central planner (representing the regional government) influences a household's decision by providing solar subsidies, modeled as percentage discounts applied to the adoption costs. The planner's problem is to minimize the total subsidy cost with a constraint to ensure that the adoption level target is achieved within the adoption time target. We model the central planner's problem as a bilevel optimization where the planner sets the subsidies with full knowledge of the household's problem, and then the households make their decisions about adoption timing and product choice.\par
In addition to the new modeling features mentioned above (timing of solar adoption, household income inequality, and solar product choice), we consider the following to be the paper's technical contributions in solving the household's problem and the planner's problem described above.
(1) We offer an explicit characterization of a household's product choice decision by determining a threshold income level that divides the household population distribution into two parts: those who will choose rooftop solar and those who will opt for subscription solar. 
(2) We develop closed-form expressions for a household's adoption timing density function.
(3) For the central planner's bilevel optimization formulation, we prove that it is a convex problem in the special case of equal subsidies, and for the general case, we prove the properties of the optimal solution and offer an algorithm to determine it. We view the formulation and solution of a bilevel optimization problem where the lower-level problem is based on an optimal stopping problem for a continuous time stochastic process as a new contribution to the OM modeling literature.\par
The solutions we develop for the household's and planner's problems allow us to investigate the interaction between the planner's policy decisions and a household's responses to them. Based on numerical explorations, we contribute the following observations. 
(1) To set goals, a central planner must consider three factors simultaneously:  adoption level target, adoption time target, and subsidy budget. A shorter target time will require a higher subsidy budget, as will a higher target adoption level. For a given subsidy budget, adoption level and adoption time can be substitutes for each other.
(2) Community solar is a useful tool for planners to achieve targets at a lower cost. While it is intuitive to suggest that adding additional product choices to the rooftop solar will make the planner's problem easier, the impact of community solar goes deeper. This is so because community solar and rooftop solar attract households from the opposite ends of the income spectrum. The central planner can differentiate the subsidies it offers for both products to help it achieve its target in a timely manner.
(3) Increasing income inequality in a region, as measured by its Gini coefficient, influences the adoption decision across the population. We show that the availability of community solar as a choice can help speed the adoption process; a low subsidy level with two products can deliver faster adoption than a high subsidy level with only one product. \par
The rest of this paper is organized as follows. In Section \ref{survey}, we review the related literature. We present the notation and assumptions and then formulate the problem in Section \ref{model}.  In Section \ref{section Analysis}, we first solve a special case of our model and then numerically solve the general model.  We discuss the implications of the results and possible extensions of the model in Section \ref{section_discussion} and Section \ref{section_extension} respectively. The proofs are deferred to the online appendix.
\section{Literature Review}\label{survey}
The emergence of distributed solar energy production and government incentives to promote it has raised interesting research questions. We start with summarizing this literature. Most of these papers, however, deploy static models to address questions they are interested in. Our focus is on adding a time dimension,  which requires us to develop a dynamic model of adoption decisions. There is limited work on developing dynamic models in the context of solar energy. Below, we describe how our model builds on the existing literature and then adds to it. Finally, our methodology has parallels in the real options literature. We close this section by highlighting these connections.\par
An early question in the literature was about the design of government subsidies to incentivize adoption by users. \cite{brown2017designing} analyzed the optimal rate of compensation for excess electricity generation in a distributed generation system when electricity generated is stochastic. \cite{sunar2021net} included the wholesale market dynamics in the analysis of the benefits of the net-metering policy to the utilities assuming random electricity demand. \cite{singh2022s} studied the optimal tariff structure for the utilities to achieve certain welfare goals. \cite{agrawal2022non} investigated the economic and environmental implications of non-ownership business models by endogenizing a solar power company's business model decisions. These papers use static models to develop their insights.\par
The literature on dynamic models in renewable energy operations is limited. \cite{babich2020promoting} model electricity prices and investment cost as diffusion processes and compares the effectiveness of two different subsidy policies, feed-in tariff, and tax rebate policies. \cite{angelus2021distributed} study when and how much capacity a single customer should install with an uncertain demand process to minimize future costs. \cite{alizamir2016efficient} offers a periodic model with information diffusion and learning over time with a goal to determine changes in the government's incentive over time. Our focus on the dynamic nature of the problem is for a different purpose. We analyze how individuals' adoption timing decisions interact with the planner's horizon for achieving the region's renewable energy goals. In addition, our model incorporates features like product choice and income inequality. Our focus on these new features has required us to simplify some of the other features available in some of the above papers, most notably the modeling of utilities and other intermediaries between the central planner and the electricity customer. We discuss such simplifications and assumptions after we present our model in the next section.\par
Our research also draws on real option valuation literature. A detailed exposition on real options valuation is available in \cite{dixit1994investment}, \cite{shiryaev2007optimal}, \cite{shreve2004stochastic} and \cite{oksendal2013stochastic}. In OM literature, techniques from real options valuation have been used to address investment timing, size, and choice decisions in various domains (\cite{dangl1999investment}, \cite{kwon2010invest}, \cite{takashima2012investment}, \cite{kwon2016impact}, \cite{angelus2021distributed}). In our work, we solve an optimal investment timing sub-problem with a choice of investment and use the results to solve a non-linear optimization problem in a bilevel formulation.
\section{ Problem and Model Formulation}{\label{model}}
We develop an optimal stopping formulation for a household's decision to adopt a solar product from a set of available choices. The households are heterogeneous in their discount rates and demand rates. They weigh the cost of adoption against future savings to determine which product to adopt and when to adopt it. Section \ref{subsec_household_pop} presents how we model household heterogeneity, and Section \ref{sub_cost} lists the cost parameters. Section \ref{sub_sub_household_optimization} puts these elements together in an optimization model for the household. Finally, a central planner, usually a state, sets targets related to adoption level and timing and chooses adoption subsidy levels to achieve them as it anticipates household responses. Section \ref{sub_sub_central_planner_optimization} models the central planner's problem.
\subsection{Household Population}\label{subsec_household_pop}
We consider a region with a population of households that are heterogeneous in their income levels. Let $r \in (0, \infty)$ denote the income level of a household. The function $p$ denotes the probability density of income in the population, i.e., $p(r)   \dd r$ is the portion of the population with income in the range of $[r,r+  \dd r]$. As discussed earlier, we are interested in studying the effect of inequality on adoption decisions. We chose income, rather than wealth or consumption, as the modeling primitive to capture inequality in our model. This choice is primarily due to the availability of direct support in the literature for two links we want to capture in our model: the link between discount rates and income levels and the link between electricity consumption and income levels.  We discuss these two features of our model below.
\subsubsection{Demand Heterogeneity}\label{subsubsec_demand_hetero}
We consider a complete filtered probability space $(\Omega, \mathcal{F}, \bbF,\bbP)$. The demand of electricity per unit time for a household belonging to income level $r$ is represented by a stochastic process, $\{ X^r_{t}\}_{t \geq 0}$. Let $\{ \mathcal{F}^r_t\}_{t \geq 0}$ be the $\sigma-$algebra generated by the stochastic process $\{ X^r_{t}\}_{t \geq 0}$, i.e., $ \mathcal{F}^r_t = \sigma(X^r_{s}: 0 \leq s\leq t)$. The filtration $\bbF= \sigma(\bigcup\limits_r \bigcup\limits_{t} \mathcal{F}^r_t$). We assume that $X^{r}_t$ follows the following dynamics
	    \begin{align}
		   \dd X^{r}_t &= \mu(r) X^r_{t} \dd t + \sigma X^r_{t} \dd W_t^{r}. \label{eq: demand process}
	    \end{align}
Here, the functions $\mu: (0, \infty) \rightarrow (0,\infty)$ and parameter $\sigma$ satisfy the at-most linear growth and Lipschitz continuity conditions.\par
The choice of representing electricity demand as a stochastic process allows us to capture the inherent variability of electricity consumption over time and provides us with a modeling platform to focus on the adoption timing decisions. This choice also follows recent literature; see  \cite{angelus2021distributed} for a similar modeling approach.\par
A wide variety of sources provide empirical evidence for the observation that the electricity demand is increasing over time and that households with higher incomes consume more. For example,  a report by NREL(\cite{steinberg2017electrification}), states that the average electricity demand is expected to increase $2.6\%$ by 2050 due to the widespread electrification of homes and modes of transport. The data(\cite{WB2021}) from 1971 to 2014 on per capita electricity consumption for three different major income levels shows that the electricity consumption for a higher income level is not just larger, it also displays a steeper gradient indicating faster growth for higher-income households. We capture these empirical observations in our model by assuming $\mu(r)>0$ for all $r$ and $\mu(r)$ is a non-decreasing function of $r$. We also assume that $X_{0}= X^{r}_{0}=x$ for all $r \in (0, \infty)$ for the sake of tractability and expositional simplicity.          
\subsubsection{Discount Rate Heterogeneity}\label{subsubsec_discount_rate_hetero}
\cite{banerjee2004inequality} observes that income level should not influence investment decisions if markets work perfectly, but that is not what we see in practice.  The paper then shows how markets are inefficient and how such inefficiencies can explain the impact of inequality on investment decisions. Moving from general arguments to the specific case of household investments in renewables, \cite{ameli2015impedes}  provides evidence of such investments being less than expected and then lists a variety of factors that may be driving such underinvestment, including the fact that lower-income households may hold back from such investments. This underinvestment is sometimes called the ``energy efficiency gap"  (\cite{gillingham2009energy}).  Investigations into why such a gap exists tend to focus on the role of the implicit discount rate; see, for example, \cite{sanstad2006end}, which describes it as the rate at which subjects discount the returns to energy-efficiency investments inferred ex-post from actual purchase decisions.\par
We seek to parsimoniously capture the effect of income differences on the adoption investment decision in our model. There may be many underlying drivers for how income levels influence household decisions -- such as credit availability, liquidity constraints, and hidden costs for low-income households --  but the above literature suggests that an effective way to capture their combined effect is to posit that the implicit discount rate which households use to make their decision is dependent on their income level.\par
In energy policy literature, the relationship between discount rates and income has long been a subject of interest. \cite{hausman1979individual} is one of the earliest works to model a household purchase decision of an energy-using durable and estimate its parameters; the results show that the discount rate varies inversely with income. This early literature is summarized in \cite{train1985discount}, which reviewed results from more than 12 studies that show discount rates decreasing with income. \cite{lawrance1991poverty} estimates that discount rates are three to five percent higher for households with low permanent incomes than for those with high permanent incomes. Newer studies like \cite{enzler2014subjective} continue to confirm the negative correlation between discount rates and income. A recent review, \cite{kubiak2016decision}, graphically represents the results of 13 studies,  further observing a correlation between high-income levels and low discount rates.\par
We consider this evidence strong enough to assume a similar relationship in our model. The income-dependent time preference discount rate is given by $\lambda(r)$. The function $\lambda$ is assumed to be continuous, bounded, and decreasing in $r$. It satisfies $\lambda(r) > \mu(r)$  and $\lambda(r) \in [\underline{\lambda}, \overline{\lambda}]$ for all $r$.
\subsection{Costs}\label{sub_cost}
From the perspective of a household considering the adoption of a solar product, we model the following costs: the price of electricity consumption, any adoption costs related to the acquisition of solar products, and net compensation after obtaining the solar product.\par
Let us first consider the consumption costs. Before adoption, a household pays for the electricity \textit{consumption} at the retail price of $p_b$ per unit. If it adopts a solar product, the cost is determined by what is generally known as the net-metering mechanism in the literature. Under net-metering, the central planner compensates the household for any electricity it delivers to the grid.  The billing cycle of the compensation mechanism, usually a month or a year, is denoted by $t_b$. At the end of a cycle, if a household’s electricity demand exceeds the total solar energy generated during a billing cycle, the household pays for the surplus consumption at retail price $p_b$ per unit, and if a household generates more solar electricity than the total demand during a billing cycle, the household is compensated for excess production at the same rate $p_b$. We refer to a household's net cost as its \textit{compensation}.\par
The net-metering solar compensation mechanism is widely used in the solar adoption models in the literature; see, for example, \cite{sunar2021net} and \cite{agrawal2022non}. It is also widely used in practice. While some variations exist, many states,  including Delaware, New Hampshire, Oregon, and Minnesota(\cite{DSIRE2021}),  use compensation mechanisms similar to the one described above, crediting surplus solar generation at retail price.\par
	%~~~~~~~~~~~~~~~~~~~~~~~~~~~~~~~~~~~~~~~~~~~~~~~Review~~~~~~~~~~~~~~~~~~~~~~~~~~~~~~~~~~~~~~~
The acquisition costs incurred by the household at the time of adoption depend on the choice of solar product. We assume that the household has a choice of two solar products:  rooftop solar and subscription solar, with index $i=1$ for the rooftop and $i=2$ for the subscription solar. The former requires an upfront cost, and the latter charges a periodic fee. Including these two choices in our model reflects contemporary practice; these two products were simultaneously available in 39 states in the U.S. in 2021(\cite{NREL2021}). The capacity acquired by a household under either product choice is assumed to be the same and is denoted by $c$. The function $\eta(c, t_b)$ represents the electricity generated by solar capacity $c$ during the billing cycle of duration $t_b$. The exogenous function $\eta$ captures the innate inefficiency in solar products due to the semiconductor material used in solar cells and the seasonality effects. The notion of a household purchasing similar capacity under both choices is supported by studies such as \cite{agrawal2022non}, which report similar average capacities of residential panels under sales and non-ownership business models.\par
The rooftop solar requires a fixed payment of $K$ and a variable payment of $k$ per unit capacity. The total cost of installing rooftop solar of capacity $c$ is $K+kc$. The subscription solar requires a periodic payment of subscription cost per unit capacity,  $p_{sub}$.  The payment structure for the two products is different; rooftop solar requires a substantial upfront payment, while subscription solar requires periodic small payments. We refer to these as the \textit{adoption} costs.\par
Finally, a household may receive a subsidy from the central planner in the form of a percentage reduction in a household's adoption costs. We model this subsidy as  $\delta_i$  where the subscript $i \in \{1,2\}$ refers to rooftop and subscription solar products, respectively. We will discuss the design of this subsidy further in Section \ref{sub_effect_subsidy_on_household}. We summarize the notations described above in a table available in the online appendix. \par
\subsection{Model Formulation}\label{sub_model_form}
We present a bilevel formulation to derive the optimal subsidy policy for the central planner. The problem naturally fits in the bilevel optimization framework since the household's solar adoption decisions across the population impact the central planner's ability to achieve the targets, and the central planner can influence the household's decisions by setting the subsidy levels. Once the central planner sets the subsidy policy, the households react to these values and either choose one of the solar products or do not adopt solar technology. In our model, the central planner is the leader, minimizing the average discounted cost due to subsidies, and the households in the region are the followers,  minimizing their expected total discounted consumption, adoption, and compensation costs. We will refer to the lower level (respectively upper level) problem as the household's (respectively central planner's) optimization problem. In the following, we first present a household's lower-level problem and then the central planner's upper-level problem.  
\subsubsection{Household's optimization problem}\label{sub_sub_household_optimization}
Consider a household at income level $r$  whose electricity consumption per unit time follows the stochastic process $X^r=(X^{r}_{t})_{t \geq 0}$. For now, we do not specify the dynamics for $X^r$;  we just assume that $X^r$ is a Markov process. For a given subsidy level $\delta$, the household must decide (i) if and when it should adopt solar technology and (ii) which of the two products it should adopt: rooftop or subscription solar.  Suppose the household chooses to adopt solar technology when its rate of electricity consumption is $x$. Then, its expected discounted future cost is given by
	  \begin{align}
		  g(x ;\delta ,r)=&  \bbE_x \bigg[  \sum_{n=1}^{\infty}   e^{-n \lambda(r)  t_b }   p_b \bigg(  \int_{(n-1)t_b}^{n t_b}X^{r}_{s} ds-\eta(c, t_b) \bigg)\bigg] \\
		  & + \min \big(  (1-\delta_1)(K + kc) , (1-\delta_2)\sum_{n=0}^{\infty} e^{-n \lambda(r)  t_b } p_{sub}c  \big). \label{eq: terminalcostbasemodel} 
   	\end{align}
   	where $\bbE_x$ denotes the expectation under $\bbP$ given $X_0^r = x$.\par
The first term represents the expected discounted cost of electricity consumption. The integral, $\int_{(n-1)t_b}^{n t_b}X^{r}_{s} ds$, is the total electricity consumption in a billing cycle, and the function, $\eta(c, t_b)$ is the electricity generated by solar capacity $c$ during the billing cycle. In any billing cycle of length  $t_b$, if the total demand is more than the solar production with acquired capacity $c$, then the household pays the rate $p_b$ for the excess electricity it consumes from the grid. Otherwise, the household receives the same rate for the excess electricity it produces and puts into the grid. This follows the net-metering policy described above. The index $n$ counts the number of billing periods forever into the future, and $\lambda(r)$ represents the discount rate used to compute the present value at the time of adoption. Irrespective of which option is adopted, rooftop or subscription, the first term remains the same due to similar capacity offerings in both the products and compensation mechanism i.e., net-metering. If we have solar products with different capacities or different compensation mechanisms, the minimum function would be outside the expectation function. The assumptions of similar capacity and compensation mechanism among solar products are to keep the model analytically tractable. The second term represents the adoption cost and depends on the chosen option.  The first term under the $\min$ function is the cost of acquiring rooftop solar, and the second term represents the present value of the stream of subscription fee payments incurred at the end of each billing cycle. The household picks the minimum of the two costs.\par
We are now ready to formulate the household's optimization problem.  We denote a household's adoption time given income level $r$ by $\tau_r:= (\tau| r)$. The random variable $\tau_r$ is adapted to the filtration $ \{ \mathcal{F}^{r}_t\}_{t \geq 0}$ generated by the demand process for income level $r$, $\{ X_t^{r}\}_{t \geq 0}$. The optimal stopping time for a household, given the income level $r$, is denoted by $\tau^*_r$ , and it is a lower-level variable. The optimal cost function for a household of a given income level $r$ and subsidy policy $\delta$, $V(x; \delta, r)$, is given by the following optimal stopping problem.
	\begin{align}
		J(x, \tau_r ; \delta, r):=&   \bbE_x \bigg[ \int_{0}^{\tau_r} e^{-\lambda(r) s} p_b X^{r}_{s}ds  + e^ {-\lambda(r) \tau_r} g(X^{r}_{\tau_r}; \delta, r) \bigg]  , \label{eq: optimalstoppingbasemodel} \\
		V(x; \delta, r) :=&  \inf_{\tau_r \geq 0} J(x, \tau_r; \delta, r),\\
		\tau^*_r:=& {\arg\min}_{\tau_r \geq 0}J(x, \tau_r; \delta, r).
	\end{align} \par
In the above optimal stopping formulation, the optimal cost function seeks to minimize the discounted sum of the running cost of electricity consumption till the stopping time and the post-adoption costs. In essence, a household incurs the running cost of electricity consumption before adoption.  It has the option to adopt now and receive a terminal payoff equal to the sum of the adoption cost and the present value of future electricity consumption, adjusted for net metering. Or it has the option to wait and delay the adoption. The household determines an optimal stopping time to adopt one of the products. The formulation above will determine when it should adopt and, when it does, which of the two options it should choose. 
\subsubsection{Central planner's optimization problem}\label{sub_sub_central_planner_optimization}
A central planner aims to meet the adoption level target $\Lambda$ , expressed in terms of the fraction of households who adopt either the rooftop or the subscription solar product within an adoption time target horizon $T$. We consider the dual adoption level and time targets, $(\Lambda,T) $, as an exogenous input to the central planner's problem. Later in the numerical section, we explore the impact of different targets.\par
Given a target, the central planner's objective is to achieve it at the minimum expected average discounted cost of the subsidy the planner provides to the households that adopt. As formulated in the previous section, a household solves its own optimal stopping problem to determine the time to adopt the preferred solar product.  Thus, household decisions across the population impact the central planner's ability to achieve the target, and the central planner can influence a household's decisions by the subsidy offered. Assuming that the central planner has complete knowledge of the household's problem,  we formulate the central planner's problem as bilevel optimization with the level of subsidy $\delta$ as the upper-level variable. We present it below and then discuss it further.\par
We first define the adoption time of a randomly chosen household from the population, $\tau^*$, and its density function $f_{\tau^*}(t; \delta, x)$.  This metric is useful for the central planner as it represents a summary response of the households to a specific subsidy policy. Note that we can interpret the probability of $\tau^*$ not exceeding a given time-horizon as the average fraction of households that will adopt in that time horizon. We use this interpretation in the formulation below.
	\begin{align}
		f_{\tau^*}(t;\delta,x)& := \int_{0}^{\infty} 	f_{\tau^{*}_r}(t;\delta, x) p(r) dr. \\
		\bbP(\tau^* \leq t | X_{0}=x)&:= \int_{0}^{\infty} 	\bbP (\tau^{*}_r \leq t | X_{0}=x) p(r) dr. 
	\end{align} 
	The bilevel formulation for the central planner is given by
	\begin{align*}
		& \min\limits_{\delta} \bbE_{x} \bigg[\int_{0}^{\infty}  e^{-\lambda(r) {\tau^{*}_r}} \bigg( \mathbbm{1}_{\{r \leq r^{*}(\delta)\}} \delta_2 \sum_{n=0}^{\infty} e^{-n \lambda(r)  t_b } p_{sub}c+  \mathbbm{1}_{\{r > r^{*}(\delta)\} }\delta_1 (K+kc) \bigg) p(r) dr \bigg] \\
		& 	\textrm{s.t.}\\
		%&  \bbP(\tau^*(\delta) \leq \bbT| X_{0}=x) ~~~~~~~~~~~~~~~~\geq \Lambda_1  \\
		%&  \bbP(\tau^*(\delta) \leq \bbT  | r \in (0,r_l), X_{0}=x)  ~~\geq \Lambda_2 \\
		%&  \bbP(\tau^*(\delta) \leq \bbT | r \in (r_h, \infty),X_{0}=x)  \geq \Lambda_3 \\
		%& (1-\delta_1)(K+kc)- (1-\delta_2) p_{sub} c ~~\geq \epsilon\\
		%& 0 \leq  \delta \leq 1 
		&  \bbP(\tau^* \leq \bbT| r \in (r^{\Lambda}_{lb},r^{\Lambda}_{ub}) , X_{0}=x) \geq \Lambda\\
		& 	\tau^*_r:= {\arg\min}_{\tau_r \geq 0}J(x, \tau_r; \delta, r)\\
		& (1-\delta_1)(K+kc)- (1-\delta_2) p_{sub} c ~~~\geq \epsilon\\
		& 0 \leq  \delta \leq 1. 
	\end{align*} \par        
In the above formulation, the objective function minimizes the average discounted cost of providing a subsidy to a household. The choice of the household for solar product is embedded in the objective function using the indicator function. The first constraint ensures that the probability of a randomly chosen household adopting a solar product within horizon $T$ is at least ${\Lambda}$; that is, the central planner must meet or exceed his target $(\Lambda,T) $. Note that we allow this randomly chosen household to belong to any specified income interval with a lower bound $r^{\Lambda}_{lb}$  and an upper bound $r^{\Lambda}_{ub}$ . This gives us the flexibility to impose an adoption target for any chosen set of income levels. We will use this flexibility later, but for now, we focus on meeting the target across the entire population: $r^{\Lambda}_{lb}=0$ and $r^{\Lambda}_{ub}= \infty$. The second constraint corresponds to the lower-level problem and is valid for all $r$. The third constraint eliminates uninteresting subsidies like those that make the entire investment in rooftop solar less than subscription payment for a single period. Here, $\epsilon$ is a very small positive real number. The last constraint ensures that the subsidy percentage for each product is between $0$ and $1$.
\section{Analysis of the Model}\label{section Analysis}
In Section \ref{sub_analysis_household}, we solve the household's optimization problem and derive a closed-form expression for a threshold demand level and the density of optimal adoption time for each income level. In Section  \ref{sub_analysis_central_planner}, we show that the central planner's optimization problem is convex in the case of homogeneous subsidies. In the case of non-homogeneous subsidies, the problem is non-convex; we offer a computationally tractable algorithm to solve the problem.
\subsection{Analysis of the Household's Optimization Problem}\label{sub_analysis_household}
To analyze this problem, we first show that it is possible to separate the two decisions: the timing of adoption and the type of product to adopt. In the following result, we show that the decision about the type of product depends only on cost parameters and can be characterized based on the income level of a customer. Specifically, we prove the existence of a subsidy-dependent threshold income level $r^*(\delta)$ that divides the range of incomes into two distinct subsets such that if a customer adopts a solar technology, she will choose roof-top solar if $r> r^*(\delta)$ and subscription solar if $r \leq r^*(\delta)$. In the following, we state the result in terms of a discount rate threshold  $\lambda^{*}_{\delta}$  where  $r^*(\delta)$=$\lambda^{-1}(\lambda^{*}_{\delta})$ and $\lambda^{-1}(\mathord{\cdot})$ is the inverse function of the income-dependent discount rate $\lambda(r)$, which is a decreasing function of $r$. This allows us to state the result without assuming a specific form for  $\lambda(r)$ and yet provide a closed-form expression for the threshold $\lambda^{*}_{\delta}$.
\begin{lemma}
		\label{lemma:thresholde}
		Given $\lambda^{*}_{\delta}= 	-\frac{1}{t_b}\ln \bigg(  1- \frac{(1- \delta_2)p_{sub}c}{(1-\delta_1)(K+kc)}\bigg)$, \\
		1. if $\lambda^{*}_{\delta} \geq \overline{\lambda}$ then $r^*(\delta)=0$.\\
		2. if $\lambda^{*}_{\delta} \leq \underline{\lambda}$ then $r^*(\delta)=\infty$.\\
		3. if $ \underline{\lambda} < \lambda^{*}_{\delta} <\overline{\lambda}$ then there exists $r^*(\delta) \in (0, \infty)$.
	\end{lemma}
Given the monotonous nature of  $\lambda(r)$, the above result states that if $\lambda^{*}_{\delta}$, exclusively determined by the cost and subsidy parameters, is too low, all customers, irrespective of their income, will adopt  subscription solar.  If $\lambda^{*}_{\delta}$ is too high, all customers will adopt roof-top solar. In between these two extremes is a region where customers below an income threshold will adopt subscription solar, and those above will adopt rooftop solar. A clear articulation of this threshold allows us to focus on the timing of adoption for the rest of this section. A direct advantage of the previous result is that it simplifies the terminal cost function $g$ in \eqref{eq: terminalcostbasemodel}. We first present the expressions for this terminal cost function.
	\begin{lemma} 
		\label{lemma:closedform terminal cost}
		Assuming $ \underline{\lambda} < \lambda^{*}_{\delta} <\overline{\lambda}$,
		$$
		g(x;r, \delta)=
		\begin{cases}
			A(r)x - B(r) + \frac{(1-\delta_2) p_{sub}c}{1-e^{-\lambda(r) t_b}} &,  r <  r^{*}(\delta)\\
			A(r)x  -B(r)  + (1-\delta_1) (K + kc)   &,  r \geq r^{*}(\delta). 
		\end{cases}
		$$
		where,
		\begin{align}
			A(r)&= \frac{  p_b }{\mu(r)}  \frac{ (1-e^{-\mu(r) t_b})}{e^{(\lambda(r)- \mu(r))t_b}-1}. \\
			B(r)&=  \frac{ p_b }{e^{\lambda(r) t_b}-1} \eta( c, t_b).
		\end{align}	
	\end{lemma}
We are now ready to solve the optimal stopping problem. We first show that for a customer from a given income level, an adoption time exists. We prove it by defining a set of admissible discounting policies and a set of income levels such that the customers belonging to this set will not adopt a solar product. We then show that this set is empty. Then, we define and analyze the continuation region of the optimal stopping formulation and derive an expression for a threshold demand level.  
	\begin{lemma}
	   \label{doublestarr}
        For the given structure of function $g$ and subsidy policy $\delta$, \\
        (i) there exists an electricity demand threshold for a customer from an income level such that solar adoption is an optimal choice if current electricity consumption is greater than the demand threshold,\\
        (ii) the threshold is given by
		%$$
		%\overline{X}(r)=
		%\begin{cases}
		%B c(p_b t_b- (1-\delta(r))p_{sub})&,  r <  r^{*}\\
		%B\bigg(c p_b t_b -(1-\delta(r))(K+kc) \frac{e^{\lambda t_b}-1}{1- e^{-\lambda t_b T(r)}} \bigg) &,  r\geq r^{*}
		%\end{cases}
		%$$
		\begin{align}
			\overline{X}(r, \delta)= \bigg( \frac{\gamma_1(r)}{\gamma_1(r)-1} \bigg) \bigg(  \frac{f(r, \delta) - B(r)}{\frac{p_b}{\lambda(r)- \mu(r)}-A(r)}   \bigg) 
		\end{align}
		where,
		\begin{align}
			f(r, \delta)= \mathbbm{1}_{\{r \leq r^*(\delta)\}} \bigg( (1-\delta_2)\frac{p_{sub}c}{1- e^{-\lambda(r)t_b} } \bigg) + \mathbbm{1}_{\{r > r^*(\delta)\}}(1-\delta_1)(K+kc).  
		\end{align}
	\end{lemma}
The above result shows that customers will eventually adopt a solar product under an appropriate subsidy policy. But we are interested in when such adoptions occur. We address this question in two steps. The first result shows that a threshold electricity demand level exists, such that when a customer's demand exceeds this level, it is optimal for that customer to adopt a solar product. The following result provides a closed-form expression for this threshold and explicitly captures its dependence on the customer's income level and the subsidy policy. Later, in the second step, we use this threshold to determine the timing of adoption.
	%\begin{lemma}
	%	\label{lemma:thresholdx}
	%	For a given subsidy policy $\delta$, there exists a threshold demand level, $\overline{X}(r, \delta)$, such that it is optimal to adopt solar technology if $X^{r}_{t} > \overline{X}(r, \delta)$. The threshold is given by
		%$$
		%\overline{X}(r)=
		%\begin{cases}
		%B c(p_b t_b- (1-\delta(r))p_{sub})&,  r <  r^{*}\\
		%B\bigg(c p_b t_b -(1-\delta(r))(K+kc) \frac{e^{\lambda t_b}-1}{1- e^{-\lambda t_b T(r)}} \bigg) &,  r\geq r^{*}
		%\end{cases}
		%$$
		%\begin{align}
		%	\overline{X}(r, \delta)= \bigg( \frac{\gamma_1(r)}{\gamma_1(r)-1} \bigg) \bigg(  \frac{f(r, \delta) - B(r) \gamma_1(r)}{\frac{p_b}{\lambda(r)- \mu(r)}-A(r)}   \bigg) \notag
		%\end{align}
		%where,
		%\begin{align}
		%	f(r, \delta)= \mathbbm{1}_{\{r \leq r^*(\delta)\}} \bigg( (1-\delta_2)\frac{p_{sub}c}{1- e^{-\lambda(r)t_b} } \bigg) + \mathbbm{1}_{\{r > r^*(\delta)\}}%(1-\delta_1)(K+kc)  \notag
	%	\end{align}
	%\end{lemma}
The closed-form expression of the threshold is instructive, as it connects the adoption cost parameters, income, and subsidy to the adoption threshold. Note that the function$ f(r, \delta)$\ captures the upfront costs for either option, the adoption cost for the rooftop solar, and the discounted sum of regular payments for the subscription solar. As this function increases, the threshold increases too. Based on this threshold, we are now ready to state a result regarding a household's time of adoption.  As a household observes the stochastic process that governs its demand, the demand level eventually meets the above threshold. At this time, the household adopts a solar product. We present a closed-form expression for the density function of this adoption time.
 \begin{lemma}
		\label{taurdensity}
		The probability density function of $\tau^{*}_{r}$ is given by
		\begin{align}
			f_{\tau^{*}_r}(t;\delta, r)= \frac{a(r, \delta)}{\sqrt{2\pi}t^{3/2}}  e^{\frac{-(a(r, \delta)-b(r)t)^2}{2t}} 
		\end{align}
		where,
		\begin{align}
			a(r, \delta)&:=  \frac{1}{\sigma} \log\bigg(\frac{\overline{X}(r, \delta)}{x}\bigg) \\
			b(r)  &:= \frac{1}{\sigma}\bigg(\mu(r) - \frac{\sigma^2}{2}\bigg)
		\end{align}
   	\end{lemma} 
We visualize the relationships between demand threshold and probability density of adoption time at different income levels and subsidies in Figure \ref{fig:Demand threshold} and Figure \ref{fig: Probability density} respectively. Figure \ref{fig:Demand threshold} shows that the threshold decreases as the overall subsidy increases. However, the exact manner of the change in threshold depends on how that subsidy is divided between the two products. Figure \ref{fig: Probability density} shows that the density of adoption time shifts to the right for households of different income levels as subsidy and income increases. We also observe that the peaks of the density curve for low and high income households can be brought closure with proper subsidy policy.
   \begin{figure}
     \begin{minipage}{.5\textwidth}
        \includegraphics[width=1\linewidth]{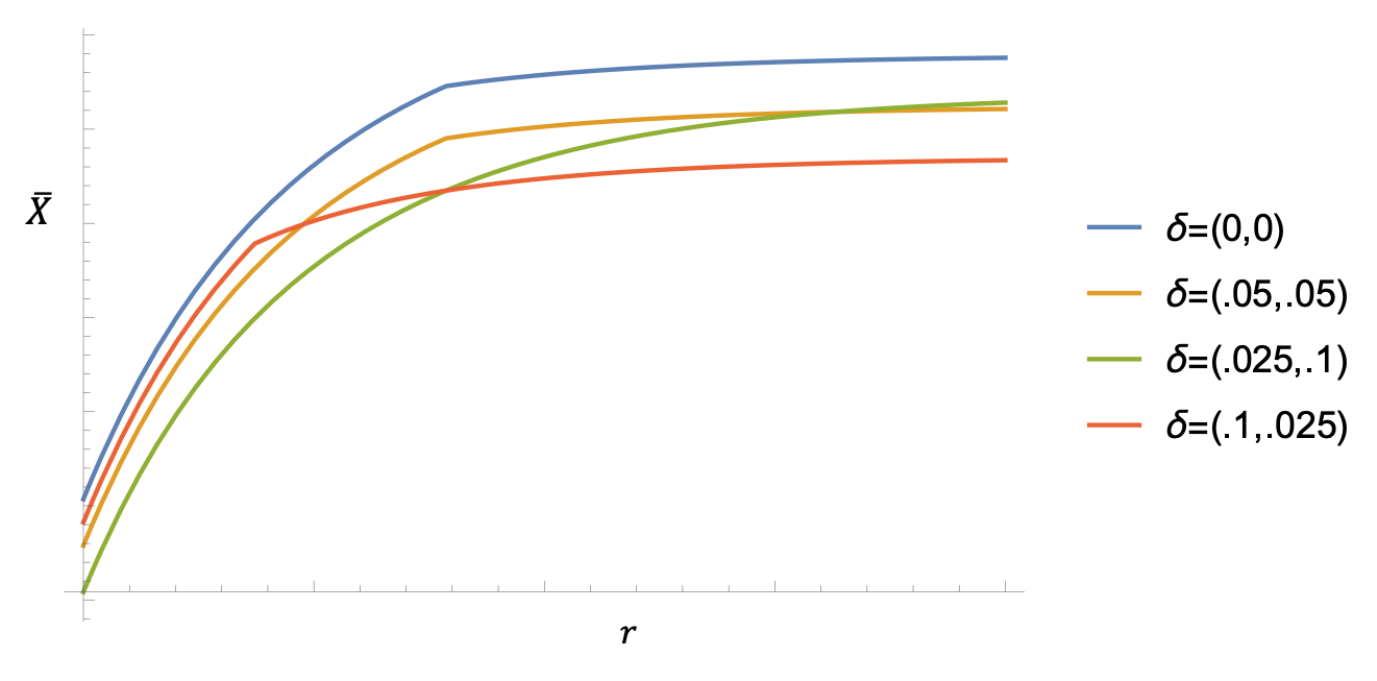}
        \captionof{figure}{$\overline{X}(r)$ for different subsidy policies}
        \label{fig:Demand threshold}
    \end{minipage}
    \begin{minipage}{.5\textwidth}
        \includegraphics[width=1\linewidth]{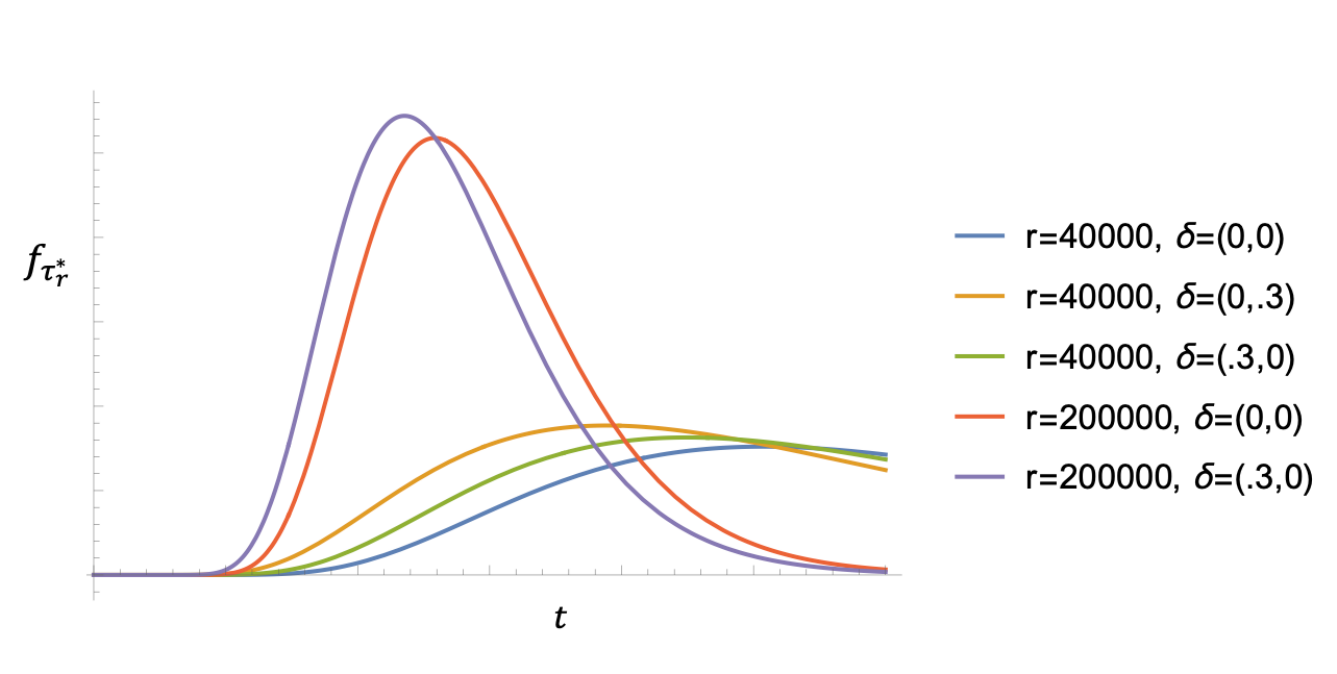}
        \captionof{figure}{Probability density of $\tau_r^{*}$}
        \label{fig: Probability density}
    \end{minipage}
    \end{figure}
To summarize this section, a household decides between the two products based on their adoption cost, income level, and the central planner's subsidy level. The product choice decision depends on comparing the upfront adoption cost of the rooftop solar product and the subscription cost stream of the subscription solar product. We show that the choice decision depends on an income-level threshold for which we provide a closed-form expression. Households from income levels higher than the threshold adopt rooftop solar, while those lower than the threshold adopt subscription solar. The adoption timing decision compares the value of waiting and the value of immediate adoption. We show that this decision, too, depends on a threshold for electricity demand, which is a function of income and subsidy levels. When a customer's demand exceeds this threshold, she finds it optimal to adopt. We provide a closed-form expression for this threshold and for the resulting adoption time density.
\subsection{Analysis of the Central Planner's Optimization Problem} \label{sub_analysis_central_planner}
In this section, we solve the bilevel formulation for the central planner presented in Section \ref{sub_sub_central_planner_optimization}. We consider two cases: homogenous subsidy and non-homogenous subsidy. In the homogenous subsidy case, we assume that the central planner assigns similar subsidies to both the products, i.e., $\delta_1=\delta_2=\delta$. Under the non-homogenous subsidy case, $\delta_1$ and $\delta_2$ can be unequal.\par
We rewrite the objective function of the central planner's problem below and refer to it as $z(\delta)$:
	\begin{align}
		z(\delta):= \bbE_{x} \bigg[\int_{0}^{\infty}  e^{-\lambda(r) {\tau^{*}_r}} \bigg( \mathbbm{1}_{\{r \leq r^{*}(\delta)\}} \delta_2 \sum_{n=0}^{\infty} e^{-n \lambda(r)  t_b } p_{sub}c+  \mathbbm{1}_{\{r > r^{*}(\delta)\} }\delta_1 (K+kc) \bigg) p(r) dr \bigg]. 
	\end{align}
\subsubsection{Homogenous subsidy}\label{sub_sub_homogenous_subsidy}
In this case, we can observe using Lemma \ref{lemma:thresholde} that $r^{*}(\delta)$ is independent of the discount policy $\delta$. This allows us to prove the following result.
	\begin{lemma} 
		\label{convexityoffeasibleregioncasefirst}
		The feasible region of the central planner's problem is a convex set.
	\end{lemma} 
	\begin {lemma}
		\label{zdelta}
		$z(\delta)$ is a convex and increasing function of $\delta$.
	\end{lemma} 
The combination of Lemma \ref{convexityoffeasibleregioncasefirst} and Lemma \ref{zdelta} shows that the central planner's problem for the case $\delta_1=\delta_2=\delta$ is a convex optimization problem. Any traditional algorithm, like gradient descent, can be used to solve the problem optimally.
	 \begin{figure}
		\centering
		\begin{subfigure}[b]{0.4\textwidth}
			\includegraphics[width=\textwidth]{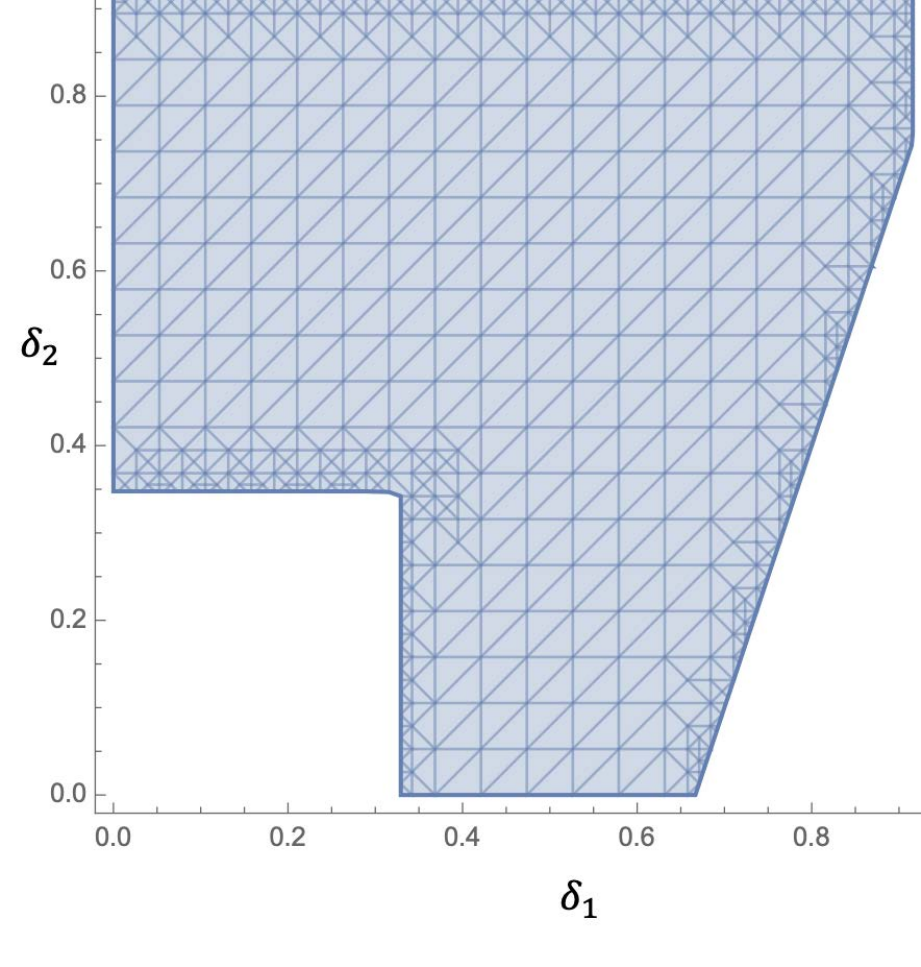} 
			\caption{}
			\label{fig non convex feasible region}
		\end{subfigure}
		\hfill
		\begin{subfigure}[b]{0.5\textwidth}
			\includegraphics[width=\textwidth]{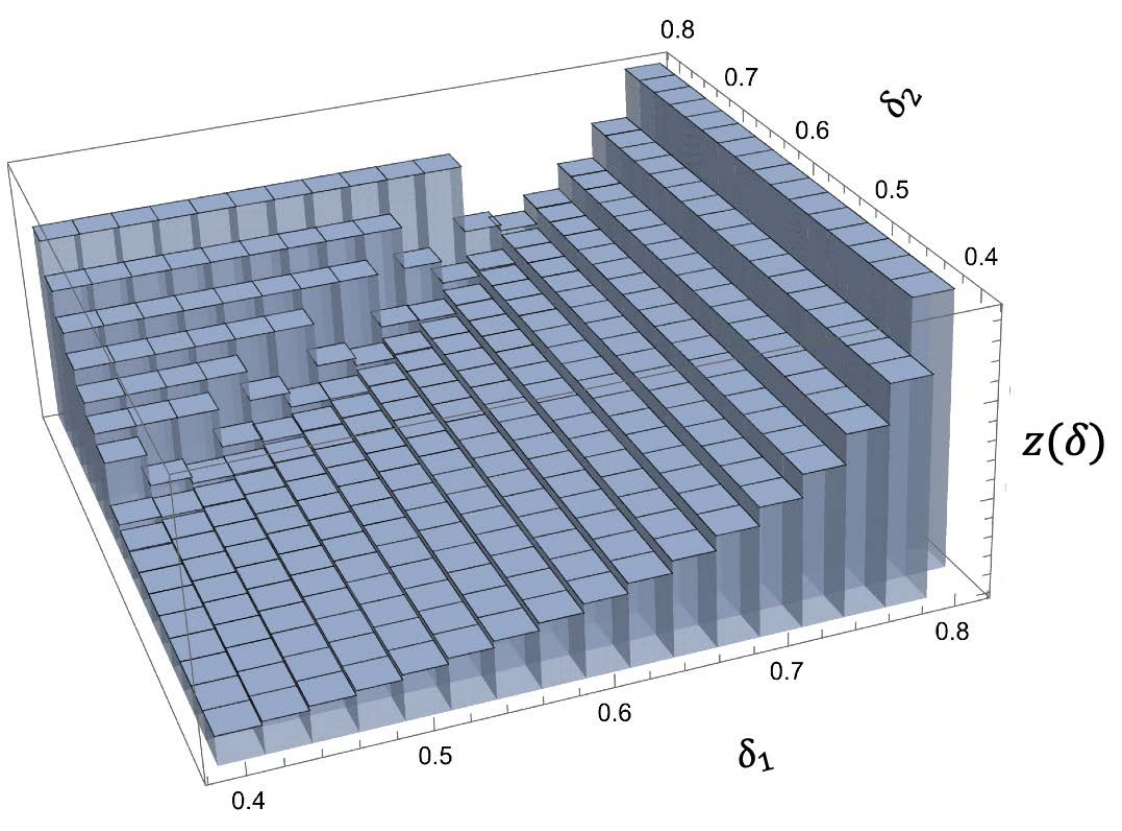} 
			\caption{}
			\label{fig Non convex obj}
		\end{subfigure}
		\caption{In case of heterogeneous subsidy (a) Feasible region (b) Objective function}
		\label{fig: Non-convex optimization problem}
	\end{figure}
\subsubsection{Heterogeneous subsidy}\label{sub_sub_non_homogenous_subsidy}
In the general case, where product subsidies can differ, the central planner's formulation is a non-convex optimization problem.  In Figure \ref{fig: Non-convex optimization problem}, we show results for a specific case of parameters when the feasible region and objective function are non-convex.The finding of non-convexity leads us to observe that a heuristic rule that simply combines two subsidy levels, each individually feasible, may not even provide a solution that satisfies all constraints.\par 
We use a grid search algorithm to find the optimal subsidy, $\delta^*$, which minimizes $z(\delta)$. We set a precision level, $prec$, and divide the $[0,1]^2$ space into disjoint grids of size $prec$ x $prec$. We choose a sequence of equally spaced points $(\delta_1^i, \delta_2^i)$ from the disjoint smaller grid, which satisfies the constraints of the central planner's optimization problem and evaluates $z(\delta)$ at each point. The pair $(\delta_1^i, \delta_2^i)$, which gives the lowest value of $z(\delta)$ is the optimal subsidy level. We can increase the accuracy of the algorithm by choosing a lower value of $prec$ and making the divisions finer, but it would increase the computational burden. We note that $z(\delta)$ is not Lipschitz continuous. Therefore, it is difficult to develop bounds for the algorithm.
\section{Discussion of the Impact of Parameters}\label{section_discussion}
After deriving theoretical results in the previous section, we now turn to developing insights into the impact of input parameters and what such insights may mean for the central planner. We first consider the role target level and target horizon play. Next, we analyze the impact of subsidy structure and what it means for the customers' decisions. We then develop insights into the impact of income distribution in the region.\\
In this section, we use numerical experiments to graph results and draw insights. The set of input parameter values for these numerical experiments is based on a variety of industry and research reports; we briefly mention them here. We use  $p_b=11.2$ cents/kWh, which was the average price of electricity to residential customers in Washington state in July 2023 (\href{https://www.eia.gov/electricity/monthly/epm_table_grapher.php?t=epmt_5_6_a}{EIA}) and $t_b= 1$ month. We use $c= 6.35$ kW, which is the mean of the average roof-top solar capacity of 6.21 kW and the average non-ownership solar capacity of 6.48 kW in the US (\cite{agrawal2022non}). We assume $\eta(6.48,1)= 500$ kWh per month for $c= 6.35$ kW based on the energy calculator(\cite{pvwattnrel}). We set $p_{sub}=\$ 22/$kW per month$, K=\$10,000$ and $k= \$4,000$, which are in the range of prices without rebates by the government. We assume that $\underline{\mu}=.01$ and $\overline{\mu}=.04$ i.e., the annual rate of increase in electricity demand is bounded between $1\%$ and $4\%$. We use annual discounting with rates $\underline{\lambda}=.045$ and $\overline{\lambda}=.06$ . Based on \cite{angelus2021distributed}, we have $x=1.6$ kW per hour.
\subsection{Adoption Level and Time Targets}\label{sub_adoption_level_time}
As discussed earlier, a distinguishing feature of our model is the focus on the adoption level and adoption time targets for the central planner. This section addresses the impact of these targets on both the central planner and the customer. It is straightforward to see the impact of these targets on the central planner's problem. A longer adoption time target  $T$ leads to lower optimal subsidy costs. This is because a longer horizon results in a larger feasible region which will be a super-set of the original feasible region due to the monotonicity of the probability measure. Using similar arguments, we can say that the higher value of adoption level target $\Lambda$ leads to higher optimal subsidy costs. In summary, more challenging targets, either high adoption levels or short planning horizons,  increase the expected discounted subsidy cost.\par
Our numerical analysis underlines the impact on the optimal central planner cost as the adoption level and time targets change. It suggests that the two types of targets can be used as substitutes for each other. If the central planner must set a tight adoption-level target, it can still control costs by relaxing the time targets. From another perspective, it is not enough to advocate for larger adoption-level targets. Such advocacy must also emphasize time horizons for achieving them; otherwise, central planners may choose the control costs by stretching the horizon over which they promise to meet the adoption level targets. \par
But as targets change, and as the central planner optimally adjusts subsidies in response, how do they influence a household's choice, rooftop or subscription solar? The length of the adoption time and level targets impact the central planner's subsidy decisions for different products. Though these targets are not part of a household's optimization problem, they impact the household's decision through subsidies assigned by the central planner to achieve the target adoption level in the target time. We explore the impact of subsidies in the next section.\par
\subsection{Effect of Subsidy on Households' Product Choice}\label{sub_effect_subsidy_on_household}
As we discussed in Section \ref{sub_analysis_household}, the value of $r^*(\delta)$ indicates the threshold for a household's choice of different solar products in the region. A household belonging to income level $r$ prefers subscription solar if $r \leq r^*(\delta)$ and roof-top solar if  $r > r^*(\delta)$. We visualize this threshold in Figure \ref{fig: income prefernce}.  A central planner can influence this choice by changing the subsidy policy. Though if the subsidy is homogeneous, i.e., $\delta_1=\delta_2= \delta$, using Lemma \ref{lemma:thresholde}, we note that $r^{*}(\delta)$ is independent of the discount policy $\delta$  and, therefore, the central planner cannot control the product choice by altering the subsidy. Under a homogeneous subsidy policy, a change in targets does not affect the choice of a household and only influences the decision to adopt or not.\par 
In the case of heterogeneous subsidy, the central planner has an additional lever at its disposal. It can differentiate between subsidies and thus also influence a household's product choice. To better understand this lever, we study the change in threshold income level, $r^*(\delta)$, when the subsidy for one of the products is changed. We show that the threshold moves to the left as $\delta_1$ increases, i.e., there are households whose choice shifts from subscription solar to rooftop solar.  The threshold moves to the right as $\delta_2$ increases, i.e., there are some households who chose rooftop solar earlier and now prefer subscription solar. Lemma \ref{propertyrstardelta}, states this formally and confirms the intuition that a higher subsidy for a product choice will favor its adoption. 
    	\begin{lemma}
	    	\label{propertyrstardelta}
		  $r^*(\delta)$ is a non-increasing function of $\delta_1$ and a non-decreasing function of $\delta_2$.
    	\end{lemma} 
    %\begin{figure}[h]
	%	\centering
	%	\includegraphics[width=0.8\textwidth]{Iso-preference} 
	%	\caption{Effect of simultaneous subsidy change on customer preference}
	%	\label{fig: simultaneous subsidy change}
	%\end{figure}
    \begin{figure}
		\centering
		\begin{subfigure}[b]{0.45\textwidth}
			\includegraphics[width=\textwidth]{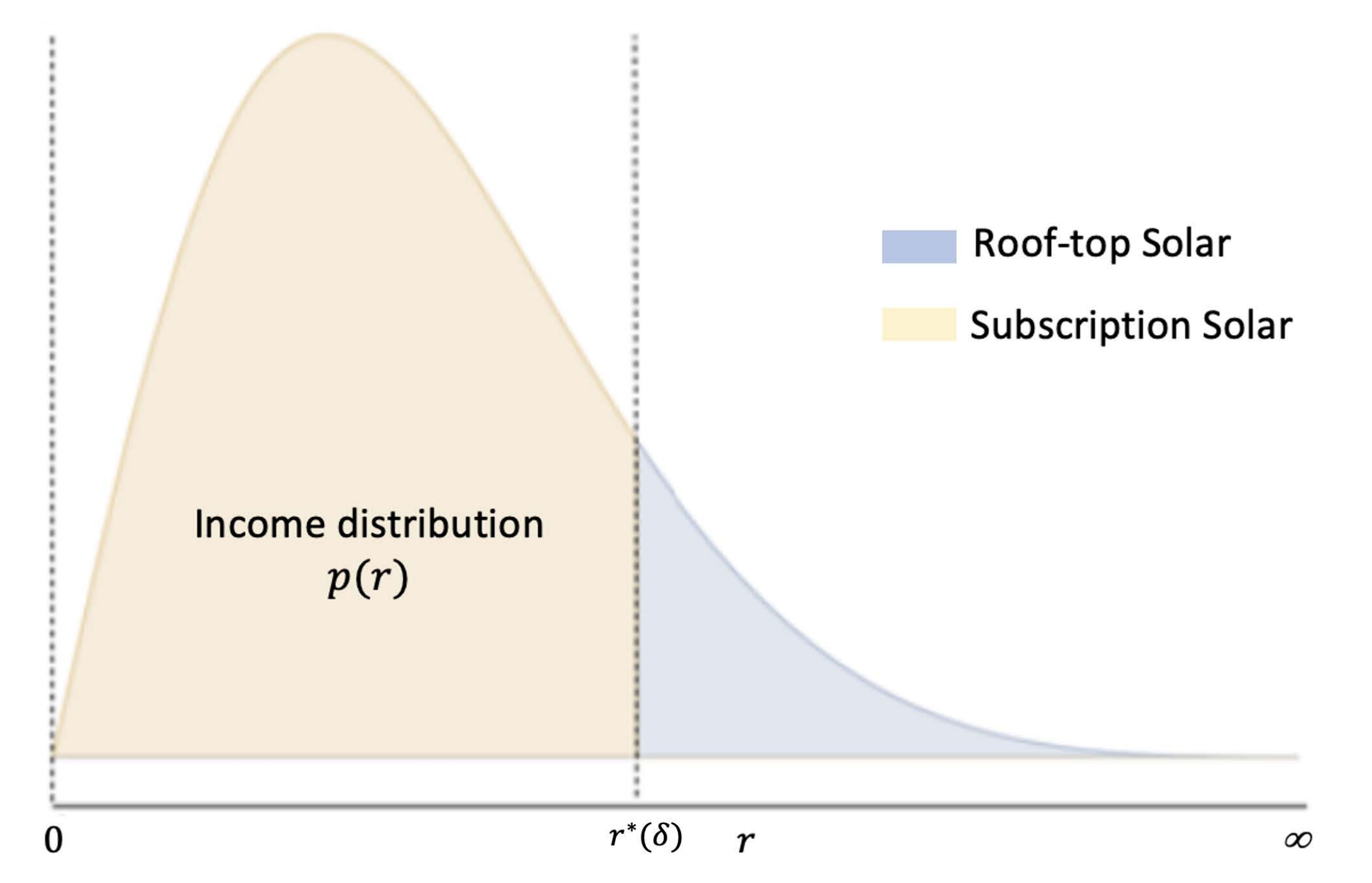} 
			\caption{}
			\label{fig: income prefernce}
		\end{subfigure}
		\hfill
		\begin{subfigure}[b]{0.5\textwidth}
			\includegraphics[width=\textwidth]{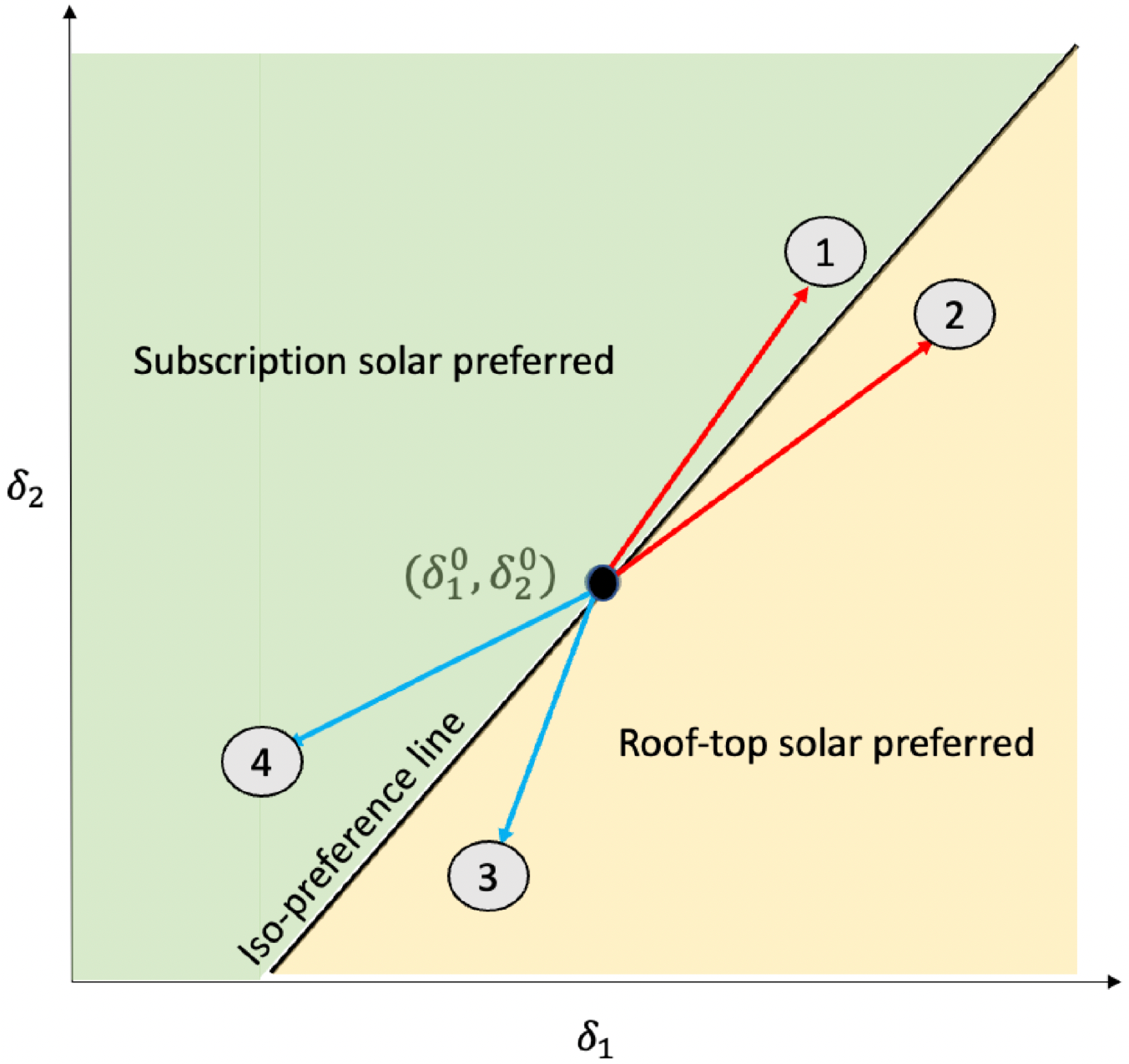} 
			\caption{}
			\label{fig: simultaneous subsidy change}
		\end{subfigure}
		\caption{(a) Product preference in the region (b) Effect of subsidy change on customer product preference}
		\label{fig: subsidy preference}
	\end{figure}
The impact on preference when subsidy for both products is changed simultaneously is not as intuitive. If both subsidies are increased (indicated by the red arrow in Figure \ref{fig: simultaneous subsidy change}) or both subsidies are decreased (indicated by the blue arrow), the impact on preference is not immediate. We now present a method that allows us to understand the impact of such subsidy changes. Let $ \delta^0 =(\delta^0_1, \delta^0_2)$ be the initial subsidy and $\delta^1=(\delta^1_1, \delta^1_2)$ be the changed subsidy such that $\delta^0_1 \neq \delta^1_1$ and $\delta^0_2  \neq \delta^1_2$. \par
We find that there exists a line passing through $(\delta^0_1, \delta^0_2)$ such that if $(\delta^1_1, \delta^1_2)$ lies on the line, then the household preferences in the region do not change. We call such a line \textit{iso-preference} line. If  $(\delta^1_1, \delta^1_2)$ lies above the iso-preference line, then preference is shifting towards subscription solar, and if $(\delta^1_1, \delta^1_2)$ lies below the iso-preference line, then preference is shifting towards roof-top solar, as shown in the Figure \ref{fig: simultaneous subsidy change}. From Lemma \ref{lemma:thresholde}, we know that 
	 \begin{align}
	    	r^{*}(\delta)= \lambda^{-1} \bigg(  \frac{-1}{t_b} \ln\bigg[ 1- q(\delta)\frac{p_{sub} c}{K+kc} \bigg]  \bigg), 
   	\end{align} 
where $q(\delta):= (1-\delta_2)/(1-\delta_1)$. In the equation above, for any $\delta^1$ such that $q(\delta^1)= q(\delta^0)$, $r^{*}(\delta^1)=r^{*}(\delta^0)$. This gives us the equation of the iso-preference line, $\delta^{1}_2= q(\delta^0)\delta^{1}_1 + 1-q(\delta^0)$. If $\delta^{1}_2> q(\delta^0)\delta^{1}_1 + 1-q(\delta^0)$ then using the decreasing property of $\lambda(\cdot)$ function we can show that $r^{*}(\delta^1)> r^{*}(\delta^0)$. Similarly, if $\delta^{1}_2 < q(\delta^0)\delta^{1}_1 + 1-q(\delta^0)$ then $r^{*}(\delta^1)< r^{*}(\delta^0)$. In Figure \ref{fig: simultaneous subsidy change}, we see that since points $1$ and $4$ lie above the iso-preference line through $\delta^0$, households'  choice would move towards subscription solar. A similar analogy can be made for points $2$ and $3$.\par
How should the central planner use the difference in subsidies to influence a household's choice, and which product should it favor: rooftop or subscription? The central planner's choice of optimal subsidies for different products depends not only on the planning horizon and adoption level but also on the income distribution in the region. We know from Lemma \ref{lemma:thresholde} that the high (low) income customers prefer rooftop (subscription) solar. Depending on the targets and the income distribution, the central planner may allocate a higher subsidy to the product preferred by the higher number of households, thus promoting early adoptions.  In our numerical experiments, when faced with a population distribution skewed toward lower incomes, the central planner finds it optimal to offer higher subsidies for community solar in order to meet tight adoption levels and time targets.  We note that this is counter to what we find in practice where planners first offered incentives for rooftop solar products.
\subsection{Income Inequality}\label{sub_income_inequality}
Another distinctive feature of this work is the explicit modeling of the income spectrum. This allows us to address questions about the impact of income distribution on solar adoption: How does income inequality impact the adoption, and who gets the lion's share of subsidy offered by the central planner? To address these questions, we assume an explicit form for the income distribution in the region. The individual income in the region follows a log-logistic distribution, i.e., $p(r)= \frac{(\beta/\alpha)(r/\alpha)^{\beta-1}}{(1+(r/\alpha)^\beta)^2}$. Here, $\alpha>0$ is a scale parameter, and $\beta>0$ is a shape parameter. According to the literature (\cite{clementi2005pareto} ,\cite{druagulescu2001evidence}), this is a well-defended choice for capturing actual income dispersion (other possible choices of income distributions are log-normal, exponential, and Pareto). In a log-logistic income distribution, the median income of the region is given by the scale parameter $\alpha$, and the Gini coefficient of the region is given by $1/\beta$. These two properties enable us to capture differences in income characteristics of different regional populations. This also provides a direct way for the central planner to incorporate income inequality in the region as it seeks an effective and fair optimal subsidy policy. We assume that $\beta>1$; this limits us to $G<1$ and makes the income distribution unimodal. The choice of log-logistic distribution is also useful because we can still derive the density function of the household adoption time, as we show in the next result.
	\begin{lemma} 
		\label{densitytau}
		The probability density function of $\tau^*$ is given by
		\begin{align}
			f_{\tau^*}(t;x, \delta,G, \alpha) = \frac{1}{\sqrt{2\pi}} \frac{1}{(G \alpha) t^{3/2}} \int_{0}^{\infty} a(r, \delta) \frac{(r/\alpha)^{\frac{1-G}{G}}}{(1+ (r/\alpha)^{1/G})^2} e^{\frac{-(a(r, \delta)-b(r)t)^2}{2t}} dr. 
		\end{align}
	\end{lemma}
    \begin{figure}[h]
		\centering
		\includegraphics[width=.8\textwidth]{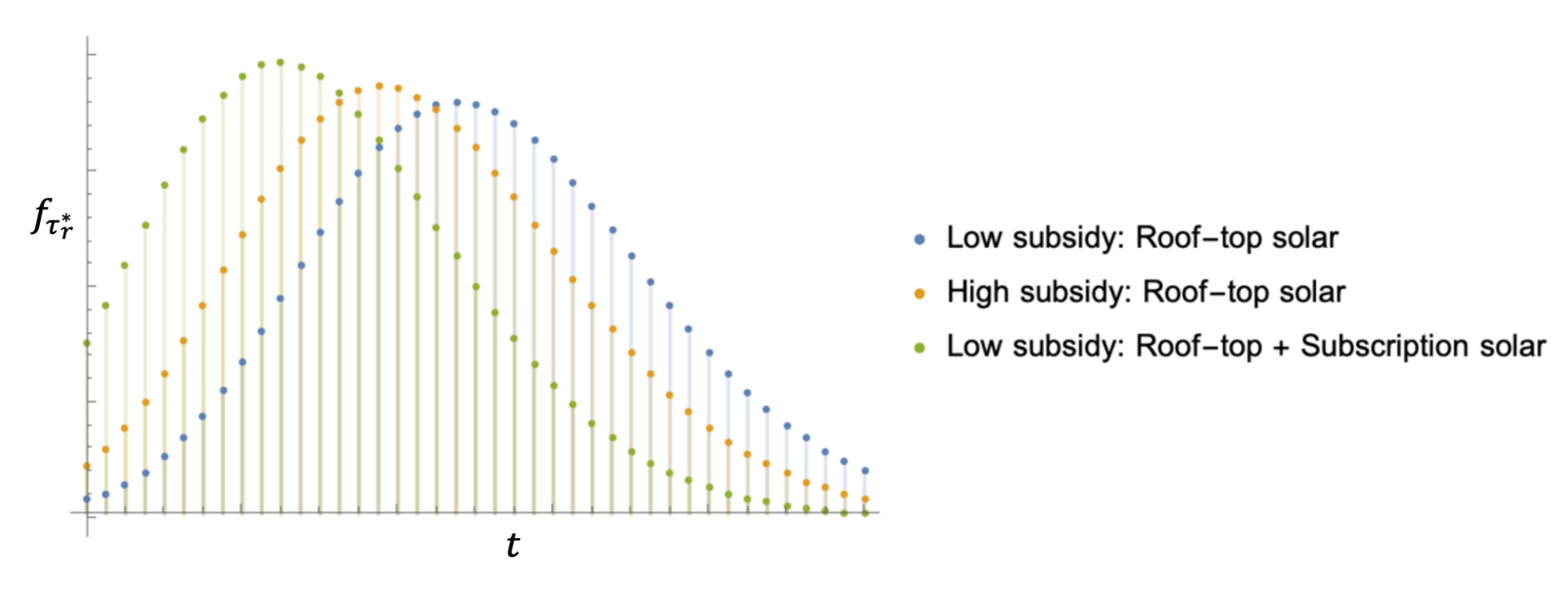} 
		\caption{Probability density function of $\tau^*$ for different product-subsidy offerings}
		\label{high income inequality product heterogeneity}
	\end{figure}
In the above Lemma \ref{densitytau}, we develop a relationship between the density function of the optimal adoption time in the region and the Gini coefficient. We can interpret the probability density function of $\tau^*$ as a weighted average of the probability density of adoption time at different income levels where the weights depend on the region's income distribution. This Lemma allows us to study the interplay between income inequality, subsidy policies, and adoption time. In Figure \ref{high income inequality product heterogeneity} we compare the effectiveness of a policy of providing households a choice of two solar products in a region of high-income inequality. The green curve represents a policy where the planner offers a low subsidy level but offers it for both products. In contrast, the orange curve represents a policy where the planner offers a higher subsidy level but  only for rooftop solar products. As the Figure shows, the mode of the adoption time distribution is smaller for the former policy than for the latter. In other words, a low subsidy level with two products can deliver faster adoption than a high subsidy level with only one product. Further reduction in subsidy for the one-product policy can delay the adoption, as is apparent from the blue curve. This observation offers useful policy advice to the central planner: offering multiple options may reduce the total subsidy cost.
\section{Extensions}\label{section_extension}
In this section, we develop several extensions of our model. We start by generalizing the net-metering credit mechanism and then add other features.
\subsection{Solar Energy Crediting Mechanisms}\label{sub_solar_crediting}
Our analysis till now assumed net-metering as the crediting mechanism for excess solar energy generation. Though a large number of states use net-metering, some states like Arizona, Kentucky, and South Carolina use net-billing. Under net-billing, the households are credited at a rate less than the retail electricity price.  In Arizona, a household is compensated at wholesale electricity price, while in Kentucky and South Carolina, the crediting price is set between retail and wholesale electricity rates. As of 2024, some states like Alabama, South Dakota, and Tennessee do not offer any form of compensation for excess solar energy generated (\cite{dsiremap}). In this section, we generalize the compensation function to incorporate net-billing and no-crediting options for excess generation along with net-metering.\par 
We denote the price for crediting excess generation by $p_s$ such that $0 \leq p_s \leq p_b$. If $p_s= p_b$, the household is compensated under the net-metering compensation mechanism. If $0 <p_s< p_b$, it is compensated under the net-billing compensation mechanism. If $p_s=0$, it can offset the electricity bill but is not compensated for the excess solar generation. \par
We define the solar compensation function $h_i$ as
     \begin{align}
	    	h_{i}(x, c; r, p_s,  p_b):=   \bbE_x \bigg[ \sum_{n=1}^{\infty}  e^{- n \lambda(r) t_b }  \bigg(  p_b  \bigg[  \int_{(n-1)t_b}^{n t_b}X^{r}_{s} ds-c t_b \bigg]^{+} - p_s \bigg[  c t_b- \int_{(n-1) t_b}^{n t_b}X^{r}_{s}ds \bigg ]^{+}  \bigg) \bigg]  .  
   	\end{align}
The function $h_i$ represents the expected payments made by a household after adopting solar product $i$. It is the sum of expected discounted payments in a billing cycle. The payments in a billing cycle depend on the demand for electricity during a billing cycle. The term $p_b[\int_{(n-1)t_b}^{n t_b}X^{r}_{s} ds-c t_b]^{+} $ is the payment made if the total demand exceeds the solar electricity generation during a billing cycle. The term $p_s [c t_b - \int_{(n-1)t_b}^{n t_b}X^{r}_{s} ds]^{+} $ is the payment received if the solar electricity generation exceeds the total demand during a billing cycle. We can re-write the above equation as
	\begin{align}
		h_{i}(x, c; r, p_s,  p_b):=  \sum_{n=1}^{\infty}  e^{- n \lambda(r) t_b } t_b  \bigg(  p_b  \bbE_x  \bigg[  \frac{\int_{(n-1)t_b}^{n t_b}X^{r}_{s} ds}{t_b}-c \bigg]^{+} - p_s  \bbE_x \bigg[  c- \frac{\int_{(n-1) t_b}^{n t_b}X^{r}_{s}ds}{t_b} \bigg ]^{+}  \bigg). 
	\end{align} 
To analyze if the household's optimal stopping problem with the above modification can be solved, we refer to the stochastic finance literature. The term $ \bbE_x  \bigg[  \frac{\int_{(n-1)t_b}^{n t_b}X^{r}_{s} ds}{t_b}-c \bigg]^{+}$ resembles the payoff of an Asian call option and the term $\bbE_x\bigg[c-\frac{\int_{(n-1) t_b}^{n t_b}X^{r}_{s}ds}{t_b} \bigg ]^{+}$ resembles the payoff of an Asian put option (\cite{shreve2004stochastic}). The household is a seller of $p_b$ units of the Asian call option, and a buyer of $p_s$ units of the Asian put option. Unlike American/European call-and-put options, there is no closed-form expression for the expected payoff of Asian options. This is primarily because the payoffs of Asian options are path-dependent due to the integral function, which makes payoffs non-Markovian. We conclude that an analytical solution for our problem may not be achievable.  However, numerous numerical methods (\cite{geman1993bessel}, \cite{rogers1995value}, \cite{vecer2001new}) exist for approximate valuation of the Asian options, which could provide better insights about different crediting mechanisms.
Even though a complete analysis of this model is not possible, it turns out that we can derive limited results for the existence of solutions for the household's optimal stopping problems in even more general settings. In the next section, we further generalize our model before presenting these results. 
\subsection{Other Features}\label{sub_other_feature}
In our primary model, the demand dynamics follow geometric Brownian motion, the \textit{compensation} to the household and various costs have fixed functional forms, and there is a choice of only two products with similar capacity. In a more general context, a household may have the option of considering different solar products with different capacities with compensation mechanisms other than net-metering, net-billing, and no solar credits. In addition, there may be other payment mechanisms available for solar adoption that require lower \textit{upfront} costs than rooftop solar but higher periodic payment than subscription solar. We extend our formulation by adding these generalized features. Specifically, we now let the demand dynamics follow a general stochastic diffusion equation.  The \textit{compensation, running, and adoption costs} can assume a general form, and the household has a choice of $n$ products with different capacities. These generalizations make detailed analytical exposition hard, but we show that exploring the edge cases is still possible. We derive a set of conditions to help the central planner devise better compensation schemes, adoption costs, and subsidy policies. We first present the model and then state the conditions.\par
Consider a complete filtered probability space $(\Omega, \mathcal{F}, \bbF,\bbP)$. The demand of electricity per unit time for a household belonging to income level $r$ is represented by a stochastic process, $\{ X^r_{t}\}_{t \geq 0}$. Let $\{ \mathcal{F}^r_t\}_{t \geq 0}$ be the $\sigma-algebra$ generated by the stochastic process $\{ X^r_{t}\}_{t \geq 0}$, i.e., $ \mathcal{F}^r_t = \sigma(X^r_{s}: 0 \leq s\leq t)$. The filtration $\bbF= \sigma(\bigcup\limits_r \bigcup\limits_{t} \mathcal{F}^r_t$). We assume that $X^{r}_t$ follows the following dynamics.
	\begin{align}
		\dd X^{r}_t &= \mu(r,X^{r}_t) \dd t + \sigma(r, X^{r}_t) \dd W_t^{r}. \label{eq: general demand process} 
	\end{align}
	Here the functions $\mu: (0,\infty)\times (0, \infty) \rightarrow (0,\infty)$ and $\sigma: (0,\infty)\times(0,\infty) \rightarrow (0,\infty)$ satisfy the at-most linear growth and Lipschitz continuity conditions.\par
The central planner offers $n$ solar products in the region. We use the notation $i$ to denote a particular product such that $i\in \{1,..,n \}$. The subsidy and capacity for products are represented by n-dimensional vectors $\delta$ and $c$, respectively. The components of these vectors, $\delta_i$, and $c_i$, are the subsidy provided by the central planner and the capacity offered for the product $i$, respectively. We assume  $\delta \in [0,1]^n$ and $c \in R^{n}_{+}$. The function $f_{i}(c_i ;r, \delta_i)$ is the \textit{adoption} cost of the solar product $i$ with capacity $c_i$ given the income level $r$ of the household and subsidy $\delta_i$ offered by the central planner. The function $h_{i}(x; r, p_s,  p_r,c_i)$ is the solar \textit{compensation} offered by product $i$ when demand for electricity is $x$, installed capacity is $c_i$, the household's income level is $r$, the retail price is $p_r$, and the re-selling price is $p_s$. We define the \textit{net cost of solar adoption} for product $i$, $g_i(x; r, \delta_i)$, as	
    \begin{align}
		g_i(x; r, \delta_i)& :=  \bbE_x\big[h_{i}(x; r, p_s,  p_r, c_i) + f_{i}(c_i;r,\delta_i)\big]. 
	\end{align}
	The function, $g(x; r, \delta)$, is the net cost of solar adoption at demand level $x$ given the income level $r$ of the household and $\delta$ subsidy policy. We define $g(x; r, \delta)$ as
	\begin{align}
		g(x; r, \delta) := \min\limits_{  i \in \{1,...n\}}  g_i(x; r, \delta_i).  \label{eq:netadoption} 
	\end{align}
	The function $w(x;p_b)$ is the cost of consuming $x$ unit of electricity per unit of time at the retail rate $p_b$. The optimal cost functional for a household, given income level $r$ and subsidy policy $\delta$, is  
	\begin{align}
		V(x; r, \delta) :=& \inf_{\tau_r \geq 0} \bigg\{ \bbE_x \bigg[ \int_{0}^{\tau_r} e^{-\lambda(r) s}w(X^{r}_{s}; p_b)ds  + e^ {-\lambda(r) \tau_r(\delta)} g(X^{r}_{\tau_r(\delta)}; r, \delta) \bigg]  \bigg\}. 
	\end{align}
	\begin{lemma}
		\label{lemma:extremestoppingtimes}
Assuming $g \in \bbC^{2}(\bbR)$, $w \in \bbC^{2}(\bbR)$ and a subsidy policy $\delta$, if  a household belongs to income level $r$ such that \\
		$(i)$ $\big( -\lambda(r) g(x;r, \delta)  + \bbL_{X^{r}} g(x;r, \delta) + w(x; p_b)\big) \geq 0$ for all $x$ then $\tau^{*}_r= 0$ a.s. \\
		$(ii)$$\big( -\lambda(r) g(x;r, \delta)  + \bbL_{X^{r}} g(x;r ,\delta) + w(x; p_b)\big) < 0$ for all $x$ then $\tau^{*}_r= \infty$ a.s. \\
		Here $\bbL_{X^r}$ is the infinitesimal generator for the demand process $\{X^{r}_t\}_{t \geq 0}$. 
		%For a given subsidy policy $\delta$, net adoption cost function $g$ and running cost function $w$,\\
		%(i) If there exists income level $\overline{r}$ such that $\big( -\lambda(\overline{r}) g(x;\overline{r}, \delta)  + \bbL_{X^{\overline{r}}} g(x;\overline{r}, \delta) + w(x)\big) \geq 0$ for all $x$ then $\tau^{*}_r= 0$ for all $r \in [\overline{r},\infty)$. \\
		%(ii)If there exists income level $\underline{r}$ such that $\big( -\lambda(\underline{r}) g(x;\underline{r}, \delta)  + \bbL_{X^{\underline{r}}} g(x;\underline{r}, \delta) + w(x)\big) < 0$ for all $x$ then $\tau^{*}_r= \infty$ for all $r \in (0, \underline{r}]$. \\
	\end{lemma}
As discussed earlier, the formulation above does not lend itself to a complete determination of adoption timing, but we are still able to identify situations in which the adoption time takes extreme values.  In Lemma \ref{lemma:extremestoppingtimes}, we derive the conditions that lead to either the immediate adoption or rejection of solar technologies by a household. It connects the net adoption cost function $g$,   running cost function $w$, and the dynamics of the demand process in \eqref{eq: general demand process} for a household with income level $r$. The result is useful for the central planner in designing net adoption cost function $g$ and the discount policy $\delta$, which prevents rejection of solar technology for any income level $r$.  Using Lemma \ref{lemma:extremestoppingtimes}, we can construct a set of admissible discounting policies and a set of income levels that will not adopt solar technology at the given discount levels.
 
 %Using Lemma (\ref{lemma:extremestoppingtimes}), we can construct a set of admissible discounting policies, $\bbD$, and a set of income levels that will not adopt solar technology at the given discount levels, $\bbH_{\delta}$ such that
	%\begin{align}
	%	\bbD & := \bigg\{  \delta \in  [0,1]^n:  \big( -\lambda(r) g(x;r, \delta)  + \bbL_{X^{r}} g(x;r ,\delta) + w(x; p_b)\big) \geq 0 \text{ for some } x \in \bbR^{+}, r \in (0, \infty) \bigg  \}, \notag \\
	%	\bbH_{\delta} &:= \bigg\{ r:  \big( -\lambda(r) g(x;r, \delta)  + \bbL_{X^{r}} g(x;r ,\delta) + w(x; p_b)\big) < 0 ~~~\forall x \in \bbR^{+}  \bigg  \}.\notag
	%\end{align}
%\textbf{ ??????????????????????????\href{}{}}
\section{Conclusion}\label{section_conclusion}
Governments and planners all around the world engage in setting targets for achieving alternative energy production. Solar energy production is a major source of alternative energy. The planners devise subsidy schemes to promote the adoption of solar products across populations. There is increasing urgency to set ambitious targets and achieve them quickly. We are motivated by a need to introduce a time dimension in the models that analyze this process: the setting of subsidies and how they influence solar adoption. We are also motivated by innovations in business models that may assist with faster adoption, specifically the emergence of community solar. Both the inclusion of the time dimension and modeling of the community solar option are unexplored topics in the solar adoption literature. We model a household's electricity demand as a continuous-time stochastic process. This choice allows us to focus on a household's adoption time decisions as well as its choice between community solar and rooftop solar. Another new feature in our model is discount rate heterogeneity, which allows us to focus on the adoption decisions across a range of income levels. Our technical contributions include a complete analysis of the household's decisions. We derive an income threshold that governs the choice of solar products and a closed-form expression for the probability density of the adoption time. We also derive technical results for the planner's bilevel optimization problem, where the planner determines subsidies to minimize its cost while meeting targets and assuming that the households will make their own adoption timing decision.\par
Numerical experiments lead to insights about target setting for the planner. Adoption level and adoption time targets behave as substitutes, allowing the planner flexibility in setting subsidies. The structure of this subsidy determines the range of household incomes that will prefer community solar over rooftop solar. Our examples show that depending on the income distribution, the planner may change its subsidy levels to achieve its adoption quantity and time targets. Finally, we show how inequality in income distributions affects the achievement of these targets.
We extend our model to address general cost and compensation functions with multiple products, but several avenues for further work remain. In future work, we will  continue to explore the effect of specific population densities and inequality measures on the design of subsidies and adoption of competing technologies.

\appendix

\section{E-Companion for Rooftop and Community Solar Adoption with Income Heterogeneity: Notations and  Proofs}

%~~~~~~~~~~~~~~~~~~~~~~~~~~~~~~~~~~~~~~~~~~~~~~~~~~~~~~~~~~~~~r(delta)=0,infinity, finite~~~~~~~~~~~~~~~~~~~~~~~~~~~~~~~~~~~~~~~~~~~~~~~~~%

  \begin{center}
      \begin{tabular}{ | m{25em} | m{1cm}| m{1cm} | } 
      \hline
       Retail cost of electricity per unit per unit time 
       & $p_b$ \\ 
      \hline
       Cost of solar subscription per unit capacity
       & $p_{sub}$ \\ 
      \hline
       Fixed cost of installing Roof-top solar
       & $K$ \\ 
      \hline
      Variable cost of installing Rooftop solar per unit capacity
      & $k$ \\ 
      \hline
      Installed solar capacity 
       & $c$ \\ 
      \hline
      Duration of solar billing cycle
      & $t_b$ \\ 
      \hline
      Electricity produced by capacity $c$ during $t_b$
      & $\eta(c, t_b)$ \\ 
      \hline
      Subsidy($\%$) for solar product $i$
      & $\delta_i$ \\ 
      \hline
   \end{tabular}
 \end{center}

\setcounter{lemma}{0}

		\begin{lemma}
		Given $\lambda^{*}_{\delta}= 	-\frac{1}{t_b}\ln \bigg(  1- \frac{(1- \delta_2)p_{sub}c}{(1-\delta_1)(K+kc)}\bigg)$, \\
		1. if $\lambda^{*}_{\delta} \geq \overline{\lambda}$ then $r^*(\delta)=0$.\\
		2. if $\lambda^{*}_{\delta} \leq \underline{\lambda}$ then $r^*(\delta)=\infty$.\\
		3. if $ \underline{\lambda} < \lambda^{*}_{\delta} <\overline{\lambda}$ then there exists $r^*(\delta) \in (0, \infty)$.
  	  \end{lemma}
  	   \begin{proof}
  	   	Consider a function $\Delta(r):= (1-\delta_1)K+kc- (1-\delta_2)\sum_{n=0}^{\infty} e^{-n \lambda(r)  t_b} p_{sub}c$.
  	   	\begin{align}
  	   		\frac{\partial \Delta(r)}{\partial r}= (1-\delta_2)\frac{e^{-\lambda(r) t_b}}{(1-e^{-\lambda(r) t_b)^2}} p_{sub} c t_b \frac{\partial \lambda(r)}{\partial r}
  	   	\end{align}
  	   	
  	   	We notice that $\Delta(r)$ is a continuous decreasing function of $r$, since $\frac{\partial \lambda(r)}{\partial r}<0$. This implies that the difference in cost of roof-top installation and subscription solar decreases in $r$. Next, we discusses 3 possible cases. \\\\
  	   	$\textbf{Case 1:}$ If $\lambda^{*}_{\delta}> \overline{\lambda}$, then
  	   	\begin{align}
  	   		-\frac{1}{t_b}\ln \bigg(  1- \frac{ (1-\delta_2)p_{sub}c}{(1-\delta_1)(K+kc)}  \bigg)  > \overline{\lambda}. \\
  	   		(1-e^{-\overline{\lambda} t_b}) < \frac{(1-\delta_2) p_{sub}c}{ (1-\delta_1) (K+ kc)}.\\
  	   		(1-\delta_1)(K+kc) - \frac{(1-\delta_2)p_{sub}c}{1-e^{-\overline{\lambda} t_b}}  < 0.
  	   	\end{align}  
  	   	The function $\Delta$ is decreasing and continuous  in $r$ and $\lambda$ is decreasing continuous bounded function of $r$. Since $\lambda(r) \leq \overline{\lambda}$, the following holds for all $r$.
  	   	\begin{align}
  	   		\Delta(r) \leq (1-\delta_1)(K+kc)  - (1-\delta_2)\frac{p_{sub}c}{1-e^{-\overline{\lambda} t_b}}  < 0
  	   	\end{align}
  	   	and $r^*(\delta)=0$.\\\\
  	   	$\textbf{Case 2:}$ If $\lambda^{*}_{\delta}< \underline{\lambda}$, then
  	   	\begin{align}
  	   		-\frac{1}{t_b}\ln \bigg(  1- \frac{(1-\delta_2)p_{sub}c}{(1-\delta_1)(K+kc)}  \bigg)  < \underline{\lambda}. \\
  	   		(1-e^{- \underline{\lambda} t_b}) > \frac{(1-\delta_2) p_{sub}c}{(1-\delta_1)(K+kc)}.\\
  	   		(1-\delta_1)(K+kc) - \frac{(1-\delta_2) p_{sub}c}{1-e^{- \underline{\lambda} t_b}}  > 0.
  	   	\end{align}  
  	   	The function $\Delta$ is decreasing and continuous  in $r$ and $\lambda$ is decreasing continuous bounded function of $r$. Since $\lambda(r) \geq \underline{\lambda}$, the following holds for all $r$.
  	   	\begin{align}
  	   		\Delta(r) \geq (1-\delta_1)(K+kc) - (1-\delta_2)  \frac{p_{sub}c}{1-e^{-\underline{\lambda} t_b}}  > 0
  	   	\end{align}
  	   	and $r^*(\delta)=\infty$.\\\\
  	   	$\textbf{Case 3:}$ If $\underline{\lambda} \leq \lambda^{*}_{\delta} \leq \overline{\lambda}  $, then there exist $r'$ and $r''$ such that $\Delta(r')<0$ and $\Delta(r'')>0$. Using intermediate value theorem, there exist a $r^*(\delta)$ such that $\Delta(r^*(\delta))=0$. Using the monotonicity of $\Delta$ function we have $\Delta(r) \geq 0$ for $r \leq r^*(\delta)$ and $\Delta(r) \leq 0$ for $r \geq r^*(\delta)$. 
  	   \end{proof}

  	     %~~~~~~~~~~~~~~~~~~~~~~~~~~~~~~~~~~~~~~~~~~~~~~~~~~~~~~~~~~~~~Terminal cost~~~~~~~~~~~~~~~~~~~~~~~~~~~~~~~~~~~~~~~~~~~~~~~~~%
  	     \begin{lemma} 
  	     	Assuming $ \underline{\lambda} < \lambda^{*}_{\delta} <\overline{\lambda}$,
  	     	$$
  	     	g(x;r, \delta)=
  	     	\begin{cases}
  	     		A(r)x - B(r) + \frac{(1-\delta_2) p_{sub}c}{1-e^{-\lambda(r) t_b}} &,  r <  r^{*}(\delta)\\
  	     		A(r)x  -B(r)  + (1-\delta_1) (K + kc)   &,  r \geq r^{*}(\delta).
  	     	\end{cases}
  	     	$$
  	     	where,
  	     	\begin{align}
  	     		A(r)&= \frac{  p_b }{\mu(r)}  \frac{ (1-e^{-\mu(r) t_b})}{e^{(\lambda(r)- \mu(r))t_b}-1}.\\
  	     		B(r)&=  \frac{ p_b  }{e^{\lambda(r) t_b}-1} \eta(c, t_b).
  	     	\end{align}	
  	     	\label{gfunction}
  	     \end{lemma}
  	     \begin{proof}
  	     	Consider the case when $r  \geq  r^*(\delta)$
  	     	\begin{align}
  	     		g(x;r, \delta)&= \bbE_x \bigg[  \sum_{n=1}^{\infty}   e^{-n \lambda(r)  t_b }   p_b  \bigg(  \int_{(n-1)t_b}^{n t_b}X^{r}_{s} \texttt{d}s-\eta(c, t_b) \bigg)\bigg] +  (1-\delta(r))(K + kc).\\
  	     		&= \sum_{n=1}^{\infty}   e^{-n \lambda(r)  t_b }   p_b  \bigg(  \int_{(n-1)t_b}^{n t_b}\bbE_x \big[  X^{r}_{s} \big] \texttt{d}s-\eta(c, t_b)\bigg)+  (1-\delta(r))(K + kc).	~~~~~~~~~~~~~~~~~(\text{Using Tonelli Theorem})\\
  	     		&= \sum_{n=1}^{\infty}   e^{-n \lambda(r)  t_b }   p_b  \bigg(  \int_{(n-1)t_b}^{n t_b} xe^{\mu(r) s}\texttt{d}s-\eta(c, t_b) \bigg)+  (1-\delta(r))(K + kc).	~~(X^{r}_{s} \text{ is a geometric Brownian motion})\\
  	     		&= \sum_{n=1}^{\infty}   e^{-n \lambda(r)  t_b }   p_b \bigg(  \frac{x}{\mu(r)}(e^{n \mu(r) t_b}- e^{(n-1) \mu(r) t_b})-\eta(c, t_b) \bigg)+  (1-\delta(r))(K + kc).\\
  	     		&= \frac{  p_b x}{\mu(r)} \sum_{n=1}^{\infty}   e^{-n \lambda(r)  t_b }    (e^{n \mu(r) t_b}- e^{(n-1) \mu(r) t_b})-p_b \eta(c, t_b) \sum_{n=1}^{\infty} e^{-n \lambda(r)  t_b }    +  (1-\delta(r))(K + kc).\\
  	     		&= \frac{  p_b x}{\mu(r)} (1-e^{-\mu(r) t_b})\sum_{n=1}^{\infty}   e^{n (\mu(r)-\lambda(r))  t_b }  -p_b \eta(c, t_b) \sum_{n=1}^{\infty} e^{-n \lambda(r)  t_b }   +  (1-\delta(r))(K + kc).\\
  	     		&= \frac{  p_b }{\mu(r)}  \frac{ (1-e^{-\mu(r) t_b})}{e^{(\lambda(r)- \mu(r))t_b}-1} x -  \frac{p_b}{e^{\lambda(r) t_b}-1}\eta(c, t_b)   +  (1-\delta(r))(K + kc).
  	     	\end{align}
  	     	Now, for the case when $r  < r^*(\delta)$,
  	     	\begin{align}
  	     		g(x;r, \delta)&= \bbE_x \bigg[  \sum_{n=1}^{\infty}   e^{-n \lambda(r)  t_b }   p_b  \bigg(  \int_{(n-1)t_b}^{n t_b}X^{r}_{s} \texttt{d}s-\eta(c, t_b) \bigg)\bigg] +  (1-\delta(r))\sum_{n=0}^{\infty} e^{-n \lambda(r)  t_b } p_{sub}c.\\
  	     		&= \frac{  p_b x}{\mu(r)} (1-e^{-\mu(r) t_b})\sum_{n=1}^{\infty}   e^{n (\mu(r)-\lambda(r))  t_b }  -p_b  \eta(c, t_b) \sum_{n=1}^{\infty} e^{-n \lambda  t_b }   +    (1-\delta(r))p_{sub}c\sum_{n=0}^{\infty} e^{-n \lambda(r)  t_b }.\\
  	     		&= \frac{  p_b }{\mu(r)}  \frac{ (1-e^{-\mu(r) t_b})}{e^{(\lambda(r)- \mu(r))t_b}-1} x  -  \frac{p_b }{e^{\lambda(r) t_b}-1} \eta(c, t_b) +  \frac{(1-\delta(r))  p_{sub}c}{1-e^{-\lambda(r) t_b}} .
  	     	\end{align}
  	     \end{proof}		
  	     
  	     %~~~~~~~~~~~~~~~~~~~~~~~~~~~~~~~~~~~~~~~~~~~~~~~~~~~~~~~~~~~~~Exitance Demand Threshold~~~~~~~~~~~~~~~~~~~~~~~~~~~~~~~~~~~~~~~~~~~~~~~~~%  	     
  	     	\begin{lemma}
  	         	For the given structure of function $g$ and subsidy policy $\delta$, \\
  	          (i) there exists a electricity demand threshold for a customer from every income level such that solar adoption is an optimal choice if current electricity consumption is greater than the demand threshold.\\
  	          (ii) the demand threshold is given by
  	         \begin{align}
  	          	\overline{X}(r, \delta)= \bigg( \frac{\gamma_1(r)}{\gamma_1(r)-1} \bigg) \bigg(  \frac{f(r, \delta) - B(r)}{\frac{p_b}{\lambda(r)- \mu(r)}-A(r)}   \bigg)
  	          \end{align}
  	          where,
  	          \begin{align}
  	          	f(r, \delta)= \mathbbm{1}_{\{r \leq r^*(\delta)\}} \bigg( (1-\delta_2)\frac{p_{sub}c}{1- e^{-\lambda(r)t_b} } \bigg) + \mathbbm{1}_{\{r > r^*(\delta)\}}(1-\delta_1)(K+kc)  
  	          \end{align}
  	          
  	       \end{lemma}
  	        \begin{proof}
  	       	(i) We use Lemma \ref{boundarystoppingtime} condition $(ii)$ to prove this lemma. For a geometric Brownian motion the infinitesimal generator is given by $ \bbL_{X^r}=  \mu(r) x \frac{\partial }{\partial x}+ \frac{\sigma^{2}}{2} x^2 \frac{\partial^2 }{\partial x^2}.$ The net solar adoption cost, $g$, is given by Lemma \ref{gfunction} and $w(x)= p_b x$.  We define $Z(r):=  -\lambda(r) g(x;r, \delta)  + \bbL_{X^r} g(x;r, \delta) + w(x) $ and consider following two cases:\\\\
  	       	\textbf{Case 1: $r< r^*(\delta)$}\\
  	       	Substituting the functional form of different functions,
  	       	\begin{align}
  	       		Z(r)=  x( (\mu(r)-\lambda(r))A(r) +p_b) - \lambda(r) \bigg(   \frac{(1-\delta_2)p_{sub}c}{1- e^{-\lambda(r)t_b}}  \bigg)
  	       	\end{align}
  	       	\textbf{Case 2: $r \geq r^*(\delta)$}\\
  	       	Similarly for this case,
  	       	\begin{align}
  	       		Z(r)=  x( (\mu(r)-\lambda(r))A(r) +p_b) - \lambda(r) (1-\delta_1)(K+kc)
  	       	\end{align}
  	       	If $(\mu(r')-\lambda(r'))A(r') +p_b > 0$, then there exists a $x$ such that $Z(r') \geq 0$ since other components of $Z(r)$ are bounded. 
  	       	\begin{prop}
  	       		$(\mu(r')-\lambda(r'))A(r') +p_b > 0$ for all $r' \in(0, \infty)$.
  	       	\end{prop}
  	       	\begin{proof}
  	       		Using the expression for $A(r')$ from Lemma \ref{gfunction} and rearranging the terms leads to the following expression:
  	       		\begin{align}
  	       			(\mu(r')-\lambda(r'))A(r') +p_b &=	-p_b\Bigg( \frac{\frac{e^{-\mu(r) t_b}-1}{-\mu(r) t_b}}{\frac{e^{(\lambda(r)-\mu(r)) t_b}-1}{(\lambda(r)-\mu(r)) t_b}}  -1\Bigg)\\
  	       			&> 0
  	       		\end{align}
  	       		The last inequality follows because the function $\frac{e^x -1}{x}$ is an increasing function and $\lambda(r)-\mu(r) > -\mu(r)$.
  	       	\end{proof}
  	       	The terminal cost function is bounded and the function $Z$ in $x$ is monotonic, we can say that $x'=\frac{\lambda(r) \bigg(   \frac{(1-\delta_2)p_{sub}c}{1- e^{-\lambda(r)t_b}}\bigg)}{(\mu(r)-\lambda(r))A(r) +p_b}$ if $r< r^*(\delta)$ and $x'= \frac{\lambda(r) (1-\delta_1)(K+kc)}{(\mu(r)-\lambda(r))A(r) +p_b}$ if $r \geq r^*(\delta)$ such that $Z(r) \geq 0$ for $x>x'$. \\ \\
  	       	
  	       	%~~~~~~~~~~~~~~~~~~~~~~~~~~~~~~~~~~~~~~~~~~~~~~~~~~~~~~~~~~~~~~~~~~~~~~~~~~~~~~~~~~~~~~~~~~~~~~~~~~~~~~~~~~~~~~%  	
  	       	(ii) We have
  	       		\begin{align}
  	       		\overline{J}(x, \tau_r ; \delta, r):=&   \bbE_x \bigg[ \int_{0}^{\tau_r} e^{-\lambda(r) s} p_b X^{r}_{s}\texttt{d}s  + e^ {-\lambda(r) \tau_r} g(X^{r}_{\tau_r}; \delta, r) \bigg]  , \notag\\
  	       		\overline{V}(x; \delta, r) :=&  \inf_{\tau_r \geq 0} J(x, \tau_r; \delta, r),\notag\\
  	       		\tau^*_r:=& {\arg\min}_{\tau_r \geq 0}J(x, \tau_r; \delta, r).\notag
  	         	\end{align}
  	       	We define
  	       	\begin{align}
  	       		\texttt{d}Y_{t}^{r} = \begin{bmatrix}
  	       			\texttt{d}t\\
  	       			\texttt{d}X_{t}^{r}
  	       		\end{bmatrix}
  	       		&=  \begin{bmatrix}
  	       			1 \\
  	       			\mu(r) X_{t}^{r}
  	       		\end{bmatrix}
  	       		\texttt{d}t + 
  	       		\begin{bmatrix}
  	       			0 \\
  	       			\sigma(r) X_{t}^{r} 
  	       		\end{bmatrix}
  	       		\texttt{d} W_{t}^{r}~~~~~~~~~~~~~~~~~~~~~~~Y_{0}^{r}=(s, x),
  	       	\end{align}
  	       	\begin{align}
  	       		\texttt{d}Z_{t}^{r} = \begin{bmatrix}
  	       			                   \texttt{d}Y_{t}^{r}\\
  	       			                   \texttt{d}U_{t}^{r}
  	       			                  \end{bmatrix}
  	       		               &=  \begin{bmatrix}
  	       		                           1 \\
  	       		                           \mu(r) X_{t}^{r} \\
  	       		                           e^{-\lambda(r)t} p_b X_{t}^{r}
  	       		                          \end{bmatrix}
  	       		                          \texttt{d}t + 
  	       		                          \begin{bmatrix}
  	       		                          	0 \\
  	       		                          	\sigma(r) X_{t}^{r} \\
  	       		                          	0
  	       		                          \end{bmatrix}
  	       		                          \texttt{d}W_{t}^{r}~~~~~~~~~~~~~~~~~~~~~~~Z_{0}^{r}=z=(s, x, u).
  	       	\end{align}
  	       	If we let  $G(z)= u + e^{-\lambda(r)s} g(x; \delta,r)$, then we 
  	       	can reduce the original problem to an equivalent problem with only terminal reward. Using section 10.3 Oksendal,
  	       	the characteristic function of $Z_t$ is given by $\bbL_{Z} G=\bbL_{X}G + \frac{\partial G}{\partial s}+e^{-\lambda(r)s} p_b x$ where $\bbL_{X}$ is the characteristic
  	       	function of $X_t$ and is given by $ \mu(r) x \frac{\partial }{\partial x}+ \frac{\sigma^{2}}{2} x^2 \frac{\partial^2 }{\partial x^2}$. We denote the continuation region of the above optimal stopping problem by $D$. We notice that the terminal cost of the optimal stopping problem is bounded above by $K+kc$ and the running cost is an increasing function of $X_{t}^{r}$ assuming $p_b > 0$.  We can observe that there exists a $\overline{X}(r, \delta)$ such that $(0, \overline{X}(r, \delta)) \subset D$. Next, we derive closed form expression for $\overline{X}(r, \delta)$ and  optimal value function using pde derived by applying separation of variables on $\bbL_{Z} G=\bbL_{X}G + \frac{\partial G}{\partial s}+e^{-\lambda(r)s} p_b x$ .\\
  	        We consider a candidate value function $V(x;r, \delta)$ then V satisfies the following equations,
  	       	\begin{align} 
  	       		& \bbL V(x;r, \delta) -\lambda(r) V(x;r, \delta) + p_b x=0 ~~~~~~~~ x \in D,\\
  	       		& V(x;r, \delta)= g(x;r, \delta)~~~~~~~~~~~~~~~~~~~~~~~~~~x \not\in D,
  	       	\end{align}
  	       	First we solve for the case $x \in D$. 
  	       	$V(x;r, \delta)$ satisfies the following non-homogeneous Cauchy-Euler differential  equation.
  	       	
  	       	\begin{align}
  	       		\mu(r) x \frac{\partial V(x;r, \delta)}{\partial x}  + \frac{1}{2} \sigma^2 x^2 \frac{\partial^2 V(x;r, \delta)}{\partial x^2} - \lambda(r) V(x;r, \delta) +  p_b x =0. 
  	       	\end{align}
  	       	The solution to the above equation is of the form
  	       	\begin{align}
  	       		V(x;r, \delta)= M_1 x^{\gamma_1} + M_2 x^{\gamma_2} + \frac{p_b x}{\lambda(r) - \mu(r)}~~~~~~~~~~~~~~~~~~~~~V(0;r,\delta)=0.
  	       	\end{align}
  	       	where, 
  	       	\begin{align}
  	       		\gamma_{1,2}(r)= \frac{-(\mu(r)- \frac{\sigma^2}{2}) \pm \sqrt{ (\mu(r)- \frac{\sigma^2}{2})^2 + 2\sigma^2 \lambda(r)}}{\sigma^2}.
  	       	\end{align}
  	       	We notice that $\gamma_1(r)>1$ and $\gamma_2(r)<0$. Since,  $\gamma_2<0$ and $V(0;r, \delta)=0$, $M_2=0$. We calculate $M_1$ and $\overline{X}(r, \delta)$ using
  	       	value matching and smooth pasting technique at the boundary of $D$.
  	       	\begin{align}
  	       		M_1 (\overline{X}(r, \delta))^{\gamma_1(r)} + \frac{p_b \overline{X}(r, \delta)}{\lambda(r) - \mu(r)}= g(\overline{X}(r,\delta);\delta, r). \label{eq: valuematching} \\
  	       		M_1 \gamma_1(r) (\overline{X}(r, \delta))^{\gamma_1(r) -1} + \frac{p_b }{\lambda(r) - \mu(r)}= g' (\overline{X}(r, \delta); \delta,r). \label{eq: smoothpasting}
  	       	\end{align}
  	       	Using the above two equations we know that $\overline{X}(r, \delta)$ satisfies the following equation
  	       	\begin{align}
  	       		\frac{(\gamma_1(r)-1)p_b \overline{X}(r, \delta)}{\lambda(r)- \mu(r)} = \gamma_1(r) g(\overline{X}(r, \delta);r, \delta) - \overline{X}(r, \delta) g' (\overline{X}(r, \delta);r, \delta). \label{eq: threshold}
  	       	\end{align}
  	       	Substituting  $g(\overline{X}(r, \delta);r, \delta)$ using Lemma \ref{gfunction} in  \eqref{eq: threshold} completes the proof.\\
  	       	We can also derive the closed form expression for $V(x; r, \delta)$ by simulatenously solving equation \eqref{eq: valuematching} and \eqref{eq: smoothpasting}.
  	       	\begin{align}
  	       		V(x; r, \delta)= \frac{(B(r)-f(r, \delta))}{\gamma_1-1} \bigg( \frac{x}{\overline{X}(r, \delta)}  \bigg)^{\gamma_1(r)} + \frac{p
  	       			_b x}{\lambda(r) - \mu(r)} \notag
  	       	\end{align}
  	       	In this procedure we have made multiple assumptions. Next, we prove some properties necessary to apply verification theorem.
  	    
  	        \begin{prop}
  	        	$\mathds{E}[\tau^{*}_{r}] < \infty$
  	        \end{prop}
  	        
  	        \begin{proof}
  	        	Given that $X^{r}_t$ follows geometric Brownian motion, we know that
  	        	\begin{align}
  	        		\log(X^{r}_{{t \wedge \tau^{*}_{r}}})&= \log(X^{r}_{t_0})+ \bigg( \mu(r)- \frac{\sigma^{2}}{2} \bigg)({t \wedge \tau^{*}_{r}}-t_0)+ \sigma (W^{r}_{{t \wedge \tau^{*}_{r}}}-W^{r}_{t_0}) \notag 
  	        	\end{align}
  	        	
  	        	Taking expectation on both sides
  	        	\begin{align}
  	        	\mathds{E}[\log(X^{r}_{{t \wedge \tau^{*}_{r}}})]&= \log(X^{r}_{t_0})+ \bigg( \mu(r)- \frac{\sigma^{2}}{2} \bigg)\mathds{E}[{t \wedge \tau^{*}_{r}}-t_0]\notag
  	        	\end{align}
  	        	Using monotone convergence theorem
  	        	\begin{align}
  	        		\mathds{E}[\tau^{*}_{r}]= \frac{\lim_{t\rightarrow \infty} \mathds{E}[\log(X^{r}_{{t \wedge \tau^{*}_{r}}})]- \log x}{ \mu(r)- \frac{\sigma^{2}}{2}} + \mathds{E}[t_0] \leq \frac{\log(\overline{X}(r, \delta))- \log x}{ \mu(r)- \frac{\sigma^{2}}{2}} + \mathds{E}[t_0] < \infty \notag
  	        	\end{align}
  	        \end{proof}
  	        Consider a stopping time $\tau'_{r}= \inf\{t: X^{r}_{t} \not\in (0, \overline{X}(r, \delta))\}$. Assuming $x < \overline{X}(r, \delta)$, clearly, $\tau^{*}_{r} \leq \tau'_{r}$. Using procedure similar to the proof of previous lemma, we can show that $\mathds{E} [\tau'_{r}]<\infty$.
  	        $V(x; r, \delta)$ is a continuously differentiable function by construction. Using Ito's lemma
  	        \begin{align}
  	        	\label{eq: expectationtaudash}
  	        	\mathds{E}[V(X^{r}_{\tau'_{r}}; r, \delta)| X^{r}_{\tau^{*}_{r}}]= & V(X^{r}_{\tau^{*}_{r}}; r, \delta) + \mathds{E} \bigg[\int^{\tau'_{r}}_{\tau^{*}_{r}}\bigg( \mu(r) x \frac{\partial V(x; r, \delta)}{\partial x}+ \frac{\sigma^{2}}{2} x^2 \frac{\partial^2 V(x; r, \delta)}{\partial x^2} - \lambda(r) V(x; r, \delta)  + p_b x  \bigg)\texttt{d}t| X^{r}_{\tau^{*}_{r}} \bigg] \notag \\
  	        	&+ \mathds{E} \bigg[ \int^{\tau'_{r}}_{\tau^{*}_{r}} \sigma \frac{\partial V(x; r, \delta)}{\partial x} \texttt{d} W_t  | X^{r}_{\tau^{*}_{r}} \bigg] 
  	        \end{align}
  	        
  	       \begin{prop}
  	       	\label{lemma: martingale}
  	       	$\mathds{E} \bigg[ \int_{\tau^{*}_{r}}^{\tau'_{r}} \sigma(r) \frac{\partial V(x; r, \delta)}{\partial x} \texttt{d} W_t  |  X^{r}_{\tau^{^*}_{r}}  \bigg]=0$
  	       \end{prop}
  	       
  	       \begin{proof}
  	       	We use Corollary 4.8 from Harrison(2013) to prove the above proposition. The corollary states, for a stopping time T, if $Z= \int_0^T X_t d W_t$ such that $X_t$ is bounded in [0,T] and $E[T]<\infty$, then $E[Z]=0$. From the previous lemma, we have that both $E[\tau^{*}_{r}]$ and $E[\tau'_{r}]$ are finite.
  	       	\begin{align}
  	       	    \frac{\partial V(x;r,\delta)}{\partial x}= M_1 \gamma_1 x^{\gamma_1 -1} + M_2 \gamma_2 x^{\gamma_2 -1} + \frac{p_b}{\lambda(r)- \mu(r)}
  	       	\end{align}
  	       	$\frac{\partial V(x;r,\delta)}{\partial x}$ is a continuous function in the compact interval $[ \overline{X}(r, \delta),  \overline{X}(r, \delta) + \Delta]$ , where $\Delta >0$. This implies that $\frac{\partial V(x;r,\delta)}{\partial x}$ is bounded. Since, $\sigma(r)$ is a known constant and $\frac{\partial V(x; r,\delta)}{\partial x}$ is bounded, we get the desired result using corollary 4.8.
  	       	\end{proof} 
  	        
  	        \begin{prop}
  	        	  \label{lemma: greaterthan0}
  	        	 $ \mathds{E} \bigg[\int^{\tau'_{r}}_{\tau^{*}_{r}}\bigg( \mu(r) x \frac{\partial V(x; r, \delta)}{\partial x}+ \frac{\sigma^{2}}{2} x^2 \frac{\partial^2 V(x; r, \delta)}{\partial x^2} - \lambda(r) V(x; r, \delta)  + p_b x  \bigg)\texttt{d} t| X^{r}_{\tau^{*}_{r}} \bigg] \geq 0$ for $x \geq  \overline{X}(r, \delta)$.
  	       \end{prop}
  	        \begin{proof}
  	        	For $x \geq  \overline{X}(r, \delta)$, $V(x; r, \delta)= g(x; r, \delta)= A(r)x -B(r)+ f(r, \delta)$. Substituting $\frac{\partial g(x; r, \delta)}{\partial x}= A(r)$ and $\frac{\partial^2 g(x; r, \delta)}{\partial x^2}= 0$ in  
  	        	$ \mu(r) x \frac{\partial V(x; r, \delta)}{\partial x}+ \frac{\sigma^{2}}{2} x^2 \frac{\partial^2 V(x; r, \delta)}{\partial x^2} - \lambda(r) V(x; r, \delta)  + p_b x $ yields
  	        	\begin{align}
  	        		& (f(r, \delta)- B(r))\bigg(   \bigg(  \frac{x}{\overline{X}(r, \delta)}\bigg) (\lambda(r)- \mu(r)) \frac{\gamma_1(r)}{\gamma_1(r) -1} - \lambda(r)    \bigg),\\
  	        		& \geq (f(r, \delta)- B(r)) \bigg(     \frac{\lambda(r) - \mu(r) \gamma_1(r)}{\gamma_1(r)-1}              \bigg) \\
  	        		& \geq 0
  	        	\end{align}
  	        	The second inequality holds as we are considering the case $x \geq  \overline{X}(r, \delta)$. We have $f(r, \delta) \geq B(r)$ because demand threshold cannot assume negative value. We can infer from the expression for $\gamma_1$ that $\gamma_1 \geq 1$ and with additional simple algebraic manipulations that $\lambda(r) - \mu(r) \gamma_1(r) \geq 0$. 
  	        \end{proof}
  	        Using results from Proposition \ref{lemma: martingale} and Proposition \ref{lemma: greaterthan0} in \eqref{eq: expectationtaudash}, we can infer that for $x \geq \overline{X}(r, \delta)$, the following inequality holds
  	         \begin{align}
  	         \mathds{E} [V(X^{r}_{\tau'_{r}}; r, \delta)| X^{r}_{\tau^{*}_{r}}] \geq V(X^{r}_{\tau^{*}_{r}}; r, \delta) 
  	         \end{align}
  	         The above inequality implies that it is better to adopt solar immediately when the level of demand hits $\overline{X}(r, \delta)$ than to wait. 
  	         \end{proof}
  	       	 %~~~~~~~~~~~~~~~~~~~~~~~~~~~~~~~~~~~~~~~~~~~~~~~~~~~~~~~Under Construction~~~~~~~~~~~~~~~~~~~~~~~~~~~~~~~~~~~~~~~~~~~~~~~~~~~~~~~~~~~~~%
   	       	 %~~~~~~~~~~~~~~~~~~~~~~~~~~~~~~~~~~~~~~~~~~~~~~~~~~~~~~~~~~~~~Adoption time density~~~~~~~~~~~~~~~~~~~~~~~~~~~~~~~~~~~~~~~~~~~~~~~~~%  	  	       	 
  	       	 	\begin{lemma}
  	       	 	\label{taurdensity}
  	       	 	The probability density function of $\tau^{*}_{r}(\delta)$ is given by
  	       	 	\begin{align}
  	       	 		f_{\tau^{*}_r(\delta)}(t;\delta, r)= \frac{a(r, \delta)}{\sqrt{2\pi}t^{3/2}}  e^{\frac{-(a(r, \delta)-b(r)t)^2}{2t}} 
  	       	 	\end{align}
  	       	 	where,
  	       	 	\begin{align}
  	       	 		a(r, \delta)&:=  \frac{1}{\sigma} \log\bigg(\frac{\overline{X}(r, \delta)}{x}\bigg) \\
  	       	 		b(r)  &:= \frac{1}{\sigma}\bigg(\mu(r) - \frac{\sigma^2}{2}\bigg)\\
  	       	 	\end{align}
  	       	 \end{lemma}
  	       	\begin{proof}
  	       	 	We first find an expression for $\bbP (\tau^{*}_r \leq \bbT | X^{r}_{0}=x)$. We define $\tilde{X}^r_{\bbT}:= \max\{ X^{r}_{s}: 0 \leq s \leq \bbT\}$. The stochastic process,  $\tilde{X^r_{t}}$, is the running maximum of the demand process $X^{r}_{t}$ till time $t$. We notice that $\bbP (\tau^{*}_r \leq {\bbT} | X^{r}_{0}=x)= \bbP(\tilde{X}^r_{\bbT} \geq \overline{X}(r)| X^{r}_{0}=x)$. We have assumed that the dynamics of $X^{r}_{t}$ follows geometric Brownian motion. The distribution of the running maximum process is given by
  	       	 	\begin{align}
  	       	 		\bbP (\tilde{X}^r_{\bbT} \geq \overline{X}(r)| X^{r}_{0}=x) &= 1- \Phi\Bigg( \frac{a(r, \delta)-b(r) \bbT}{\sqrt{\bbT}} \Bigg) 
  	       	 		+ e^ {2 a(r, \delta) b(r)} \Phi \Bigg( \frac{-a(r, \delta)- b(r) \bbT}{\sqrt{\bbT}} \Bigg). 
  	       	 	\end{align}
  	       	 	where, $\Phi(\cdot)$ is the CDF  of standard normal distribution. Using the following equation gives us the desired result.
  	       	 	\begin{align}
  	       	 		f_{\tau^{*}_r(\delta)}(t;x) := \frac{d}{d t} \bbP (\tau^{*}_r(\delta) \leq t | X^{r}_{0}=x).
  	       	 	\end{align}
  	       	 \end{proof}
   	       	 %~~~~~~~~~~~~~~~~~~~~~~~~~~~~~~~~~~~~~~~~~~~~~~~~~~~~~~~~~~~~~convexity of feasible region~~~~~~~~~~~~~~~~~~~~~~~~~~~~~~~~~~~~~~~~~~~~~~~~~%    	     
   	       	\begin{lemma} 
   	       	 	The feasible region of the central planner's problem is a convex set.
   	       	 \end{lemma}
   	       	  \begin{proof}
   	       	 	Notice that the last constraint reduces to $(1-\delta)(K+kc- p_{sub}c) \geq \epsilon$ and is true for all subsidies. We need to study the first 3 constraints. We show that the right hand side of those constraint is a quasi-concave function. The super level sets of quasi-concave functions are convex.\\
   	       	 	\begin{align}
   	       	 		\bbP(\tau^*(\delta) \leq \bbT| X_{0}=x) &= \int_{0}^{\infty} \bbP (\tau^{*}_r(\delta) \leq \bbT| X^{r}_{0}=x) p(r) \texttt{d} r \\
   	       	 		&= \int_{0}^{\infty}  \bbP (\tilde{X}^r_{\bbT} \geq \overline{X}(r, \delta)| X^{r}_{0}=x) p(r) \texttt{d} r
   	       	 	\end{align}
   	       	 	For each $r$, $P(\tilde{X}^r_{\bbT} \geq \overline{X}(r, \delta)| X^{r}_{0}=x) $ is a non-decreasing function of $\delta$ since $\overline{X}(r, \delta)$ is non-increasing in $\delta$. Monotonicity of integrals imply that $ \bbP(\tau^*(\delta) \leq \bbT| X_{0}=x)$ is non-decreasing function of $\delta$. Any monotonic function is quasi-concave and the super level set of a quasi-concave functions is convex. We can use similar argument to prove that the set of $\delta$ satisfying second and third constraint are also convex. The feasible region for the central planner's problem is the intersection of these finite number of convex sets, hence convex.
   	       	 \end{proof}

   	      %~~~~~~~~~~~~~~~~~~~~~~~~~~~~~~~~~~~~~~~~~~~~~~~~~~~~~~~~~~~~~z(\delta) is convex~~~~~~~~~~~~~~~~~~~~~~~~~~~~~~~~~~~~~~~~~~~~~~~~~%    
   	       	  \begin{lemma}
   	       	     $z(\delta)$ is a convex and increasing function of $\delta$.
   	       	  \end{lemma}
   	       	   \begin{proof}
   	       	  	Let 
   	       	  	\begin{align}
   	       	  		\nu(r):=  \mathbbm{1}_{\{r \leq r^*\}} \bigg( \frac{p_{sub}c}{1- e^{-\lambda(r)t_b} } \bigg) + \mathbbm{1}_{\{r > r^*\}}(K+kc)
   	       	  	\end{align}
   	       	  	Since the integrand in $z(\delta)$ is non-negative, using Tonelli theorem, we can exchange the expectation and the integral.
   	       	  	\begin{align}
   	       	  		z(\delta):= \int_{0}^{\infty}   \bbE_{x} \bigg[ e^{-\lambda(r) {\tau^{*}_r(\delta)}}  \bigg]  \delta \nu(r) p(r) \texttt{d} r \label{eq: zdeltaconcex}
   	       	  	\end{align}
   	       	  	The expression $\bbE_{x} \big[ e^{-\lambda(r) {\tau^{*}_r(\delta)}}  \big]$ represents the Laplace transform of the adoption time for a given $r$ and $\delta$. Using the expression of the Laplace transform of the stopping time for a geometric Brownian motion, we have
   	       	  	\begin{align}
   	       	  		\bbE_{x} \big[ e^{-\lambda(r) {\tau^{*}_r(\delta)}}  \big]= \bigg(    \frac{x}{\overline{X}(r, \delta)}  \bigg)^l \label{eq: laptau}
   	       	  	\end{align}
   	       	  	where,
   	       	  	\begin{align}
   	       	  		l(r)= \frac{    \sqrt{\mu(r)^2 + 2 \lambda(r) \sigma^2}    -\mu(r)}{\sigma^2}
   	       	  	\end{align}
   	       	  	
   	       	  	Substituting (\refeq{eq: laptau}) in (\refeq{eq: zdeltaconcex})
   	       	  	\begin{align}
   	       	  		z(\delta):= \int_{0}^{\infty}     x^{l(r)}    \frac{\delta}{(\overline{X}(r, \delta))^l}  \nu(r) p(r) \texttt{d} r 
   	       	  	\end{align}
   	       	  	The integrand in the above integral is a bounded and a continuous function of $\delta$ and $r$. Using dominated convergence theorem we can show that $z(\delta)$ is a continuous function of $\delta$.
   	       	  	Next we substitute the expression for $\overline{X}(r, \delta)$ from lemma 3, which leads us to the following two cases.\\
   	       	  	
   	       	  	\textbf{Case 1:} $\overline{X}(r, \delta)>x$
   	       	  	\begin{align}
   	       	  		z(\delta)& = \int_{0}^{\infty}        x^{l(r)} \frac{\delta}{((1-\delta) \nu(r)- B(r)\gamma_1(r))^{l(r)}}  \bigg(  \bigg( \frac{\gamma_1(r)-1}{\gamma_1(r)}\bigg)   \bigg(  \frac{p_b}{\lambda(r)-\mu(r)} - A(r)\bigg)  \bigg)^{l(r)} \nu(r) p(r) \texttt{d} r. \\
   	       	  		\frac{\partial}{\partial \delta}  z(\delta) &=  \int_{0}^{\infty}   x^{l(r)} \frac{\partial}{\partial \delta}   \bigg(  \frac{\delta}{((1-\delta) \nu(r)- B(r)\gamma_1(r))^{l(r)}} \bigg)  \bigg(  \bigg( \frac{\gamma_1(r)-1}{\gamma_1(r)}\bigg)   \bigg(  \frac{p_b}{\lambda(r)-\mu(r)} - A(r)\bigg)  \bigg)^{l(r)} \nu(r) p(r) \texttt{d} r.\\
   	       	  		\frac{\partial^2}{\partial \delta^2}  z(\delta) &= \int_{0}^{\infty}  x^{l(r)}   \frac{\partial^2}{\partial \delta^2}   \bigg(  \frac{\delta}{((1-\delta) \nu(r)- B(r)\gamma_1(r))^{l(r)}} \bigg)  \bigg(  \bigg( \frac{\gamma_1(r)-1}{\gamma_1(r)}\bigg)   \bigg(  \frac{p_b}{\lambda(r)-\mu(r)} - A(r)\bigg)  \bigg)^{l(r)}\nu(r) p(r)\texttt{d}r.
   	       	  	\end{align}
   	       	  	We observe that for every $r$
   	       	  	\begin{align}
   	       	  		\frac{\partial}{\partial \delta}  \frac{\delta}{((1-\delta) \nu(r)- B(r)\gamma_1(r))^{l(r)} } &= \frac{ ((1-\delta) \nu(r)- B(r)\gamma_1(r))^{l(r)} + \delta l(r) \nu(r) ((1-\delta) \nu(r)- B(r)\gamma_1(r))^{l(r)-1} }{((1-\delta) \nu(r)- B(r)\gamma_1(r))^{2l(r)}}\\
   	       	  		& \geq 0.
   	       	  	\end{align}
   	       	  	Similarly
   	       	  	\begin{align}
   	       	  		\frac{\partial^2}{\partial \delta^2}  \frac{\delta}{((1-\delta) \nu(r)- B(r)\gamma_1(r))^{l(r)}} &= \frac{l(r) \nu(r)}{((1-\delta) \nu(r)- B(r)\gamma_1(r))^{l(r)+1}}\\
   	       	  		& + l(r) \nu(r) \bigg[ \frac{ ((1-\delta) \nu(r)- B(r)\gamma_1(r))^{l(r)+1} + \delta (l(r)+1)  ((1-\delta) \nu(r)- B(r)\gamma_1(r))^{l(r)} }{((1-\delta) \nu(r)- B(r)\gamma_1(r))^{2l(r)+2}}   \bigg]\\
   	       	  		& \geq 0.
   	       	  	\end{align}
   	       	  	This implies that $\frac{\partial}{\partial \delta}  z(\delta) \geq 0$ and $ \frac{\partial^2}{\partial \delta^2}  z(\delta) \geq 0$.\\
   	       	  	
   	       	  	\textbf{Case 2:} $\overline{X}(r, \delta)=x$\\
   	       	  	Substituting $\overline{X}(r, \delta)=x$ in the expression for $z(\delta)$, we can notice that $z(\delta)$ is a linear function of $\delta$, $\frac{\partial}{\partial \delta}  z(\delta) \geq 0$ and $ \frac{\partial^2}{\partial \delta^2}  z(\delta) =0$.
   	       	  \end{proof}

   	       	    %~~~~~~~~~~~~~~~~~~~~~~~~~~~~~~~~~~~~~~~~~~~~~~~~~~~~~~~~~~~~~r*(delta) non-increasing fn of delta ~~~~~~~~~~~~~~~~~~~~~~~~~~~~~~~~~~~~~~~~~~~~~~~~~%  	 
   	       	    \begin{lemma}
   	       	    	$r^*(\delta)$ is non-increasing function of $\delta_1$ and a non-decreasing function of $\delta_2$.
   	       	    \end{lemma}
                
   	       	    \begin{proof}
   	       	    	We know that $\lambda^{*}_{\delta }= 	-\frac{1}{t_b}\ln \bigg(  1- \frac{(1- \delta_2)p_{sub}c}{(1-\delta_1)(K+kc)}\bigg)$. It is straight forward to see that $\frac{\partial \lambda^{*}_{\delta }}{\partial \delta_1} \geq 0$ and 
   	       	    	$\frac{\partial \lambda^{*}_{\delta }}{\partial \delta_2} \leq 0$. The result follows since $\lambda(r)$ is a decreasing function of r.
   	       	    \end{proof}

   	       	    %~~~~~~~~~~~~~~~~~~~~~~~~~~~~~~~~~~~~~~~~~~~~~~~~~~~~~~~~~~~~~probability density tau*~~~~~~~~~~~~~~~~~~~~~~~~~~~~~~~~~~~~~~~~~~~~~~~~~%  
   	       	    \begin{lemma}
   	       	    	The probability density function of $\tau^*$ is given by
   	       	    	\begin{align}
   	       	    		f_{\tau^*}(t;x, \delta,G, \alpha) = \frac{1}{\sqrt{2\pi}} \frac{1}{(G \alpha) t^{3/2}} \int_{0}^{\infty} a(r, \delta) \frac{(r/\alpha)^{\frac{1-G}{G}}}{(1+ (r/\alpha)^{1/G})^2} e^{\frac{-(a(r, \delta)-b(r)t)^2}{2t}} \texttt{d} r.
   	       	    	\end{align}
   	       	    \end{lemma}
   	       	    \begin{proof}
   	       	    	The probability density function of $\tau^*$ is defined as
   	       	    	\begin{align}
   	       	    		f_{\tau^*}(t;x, \delta, G, \alpha)& := \int_{0}^{\infty} 	f_{\tau^{*}_r}(t;x, \delta) p(r) \texttt{d} r.\notag
   	       	    	\end{align} 
   	       	    	We know from Lemma \ref{taurdensity} the closed form expression for $	f_{\tau^{*}_r}(t;x, \delta) $. Substituing the expression for $p(r)= \frac{1}{G \alpha} \frac{(r/\alpha)^{\frac{1-G}{G}}}{(1+ (r/\alpha)^{1/G})^2}$ gives us the desired result.
   	       	    \end{proof}
   	       	    
   	       	  %~~~~~~~~~~~~~~~~~~~~~~~~~~~~~~~~~~~~~~~~~~~~~~~~~~~~~~~~~~~~~Generalization ~~~~~~~~~~~~~~~~~~~~~~~~~~~~~~~~~~~~~~~~~~~~~~~~~%  
   	       	     
                	\begin{lemma}
                    Assuming $g \in \bbC^{2}(\bbR)$, $w \in \bbC^{2}(\bbR)$ and a subsidy policy $\delta$, if  a customer belongs to income level $r$ such that \\ \\
                	$(i)$ $\big( -\lambda(r) g(x;r, \delta)  + \bbL_{X^{r}} g(x;r, \delta) + w(x; p_b)\big) \geq 0$ for all $x$ then $\tau^{*}_r= 0$ a.s. \\
                	$(ii)$$\big( -\lambda(r) g(x;r, \delta)  + \bbL_{X^{r}} g(x;r ,\delta) + w(x; p_b)\big) < 0$ for all $x$ then $\tau^{*}_r= \infty$ a.s. \\ \\
                	Here $\bbL_{X^r}$ is the infinitesimal generator for the demand process $\{X^{r}_t\}_{t \geq 0}$. 
                	\label{boundarystoppingtime}
                	%For a given subsidy policy $\delta$, net adoption cost function $g$ and running cost function $w$,\\
                	%(i) If there exists income level $\overline{r}$ such that $\big( -\lambda(\overline{r}) g(x;\overline{r}, \delta)  + \bbL_{X^{\overline{r}}} g(x;\overline{r}, \delta) + w(x)\big) \geq 0$ for all $x$ then $\tau^{*}_r= 0$ for all $r \in [\overline{r},\infty)$. \\
                	%(ii)If there exists income level $\underline{r}$ such that $\big( -\lambda(\underline{r}) g(x;\underline{r}, \delta)  + \bbL_{X^{\underline{r}}} g(x;\underline{r}, \delta) + w(x)\big) < 0$ for all $x$ then $\tau^{*}_r= \infty$ for all $r \in (0, \underline{r}]$. \\
                 \end{lemma}
   	       	     \begin{proof}
   	       	     	We introduce a new Ito's diffusion $Z^{r}_t$ to define a time homogeneous stopping problem as follows
   	       	     	\begin{align}
   	       	     		d Z^{r}_t=
   	       	     		\begin{bmatrix}
   	       	     			1\\
   	       	     			\mu(r, X_{t}^{r}) \\
   	       	     			e^{-\lambda(r)t} w(X_{t}^{r})
   	       	     		\end{bmatrix} \texttt{d}t +
   	       	     		\begin{bmatrix}
   	       	     			0\\
   	       	     			\sigma(r, X_{t}^{r}) \\
   	       	     			0
   	       	     		\end{bmatrix} \texttt{d}W^{r}_t	  ~~~~~~~~~~~~~~~~~Z_0=(0,x,0).   
   	       	     	\end{align}
   	       	     	We then define the a new cost function as
   	       	     	\begin{align}
   	       	     		\tilde{g}(t,x,\overline{w}; r, \delta)=\overline{w}+ e^{-\lambda(r)t }g(x;r,\delta),
   	       	     	\end{align}	
   	       	     	and our equivalent optimal stopping problem is
   	       	     	\begin{align}
   	       	     		V(x; r, \delta) := \inf_{\tau_r \geq 0}  \bbE_{Z_0} \big[ 	\tilde{g}(Z_{\tau_r }; r, \delta) \big],  \label{eq:v(x)} 
   	       	     	\end{align}
   	       	     	where $Z_{\tau_r }= (\tau_r, x, \overline{w}_{\tau_r})$.\\ \\
   	       	     	Next, we write the infinitesimal generator $\bbL_{Z^r}$ for the process $\{Z^r_t\}_{t \geq 0}$. Let $\phi(z) \in \bbC^2$, then
   	       	     	\begin{align}
   	       	     		\bbL_{Z^r} \phi= \frac{\partial \phi}{\partial t} +  \mu(r, x)\frac{\partial \phi}{\partial x} + \frac{\sigma^2(r,x)}{2} \frac{\partial^2 \phi}{\partial x^2} +  e^{-\lambda(r)t} w(x)  \frac{\partial \phi}{\partial \overline{w}}. \label{eq: zgenerator}
   	       	     	\end{align}
   	       	     	%Let $\bbL_{X^r}$ denote the infinitesimal generator for the demand process $\{X^{r}_t\}_{t \geq 0}$. 
   	       	     	Using equations \eqref{eq:v(x)} and \eqref{eq: zgenerator}, we get
   	       	     	\begin{align}
   	       	     		\bbL_{Z^r} 	\tilde{g} & = \big( -\lambda(r) g(x;r, \delta) + \mu(r,x) g'(x;r, \delta) + \frac{\sigma^2(r,x)}{2} g''(x;r, \delta) +w(x) \big) e^{-\lambda(r)t}.\\
   	       	     		& = \big(   -\lambda(r) g(x;r, \delta)  + \bbL_{X^r} g(x;r, \delta) + w(x)  \big)  e^{-\lambda(r)t},
   	       	     	\end{align}
   	       	     	For a given $r$ and $\delta$ we consider two cases:\\\\
   	       	     	\textbf{Case 1:} If $\big( -\lambda(r) g(x;r, \delta)  + \bbL_{X^r} g(x;r, \delta) + w(x)\big) \geq 0$ for all $x$, then $\bbL_{Z^r} 	\tilde{g} \geq 0$. Since $g$ and $w$ belong to the class of twice differentiable functions, $\tilde{g} \in \bbC^2(\bbR)$. Using Dynkin formula(Oksendal, Section 10.1) we can say that $\tilde{g} $ is a sub-harmonic function and $\tau^{*}_r= 0$ i.e customers from income level $r$ adopt solar technology immediately.\\ \\
   	       	     	\textbf{Case 2:} If $ \big( -\lambda(r) g(x;r, \delta)  + \bbL_{X^r} g(x;r, \delta) + w(x) \big) < 0$ for all $x$, then $\bbL_{Z^r} 	\tilde{g} < 0$. This implies that the continuation region is $\bbR$ and $\tau^{*}_r= \infty$ i.e customers from income level $r$ do not adopt solar technology.\\
   	       	     \end{proof}	     	    %~~~~~~~~~~~~~~~~~~~~~~~~~~~~~~~~~~~~~~~~~~~~~~~~~~~~~~~~~~~~~r*(delta) non-increasing fn of delta ~~~~~~~~~~~~~~~~~~~~~~~~~~~~~~~~~~~~~~~~~~~~~~~~~%  

\bibliographystyle{informs2014} % outcomment this and next line in Case 1
\bibliography{citation.bib} % if more than one, comma separated
							
\end{document}